\documentclass[a4paper]{llncsreport}
\pdfoutput=1 %
\usepackage[utf8]{inputenc}
\usepackage[T1]{fontenc}
\usepackage[ngerman,english]{babel}
\usepackage{llncsstuffthm}
\usepackage{llncsstuffset}
\usepackage{cwdefs}
\usepackage{amssymb}
\usepackage{amsfonts}
\usepackage{multirow}
\usepackage{longtable}
\usepackage{booktabs}
\usepackage{amsmath}
\usepackage{bigdelim}
\usepackage{comment}
\usepackage{xspace}
\usepackage{booktabs}
\usepackage{graphicx}
\usepackage{setspace}
\usepackage{upgreek}
\usepackage{enumitem}
\usepackage[linktoc=all,hidelinks,colorlinks=false,hyperfootnotes=false]{hyperref}
\usepackage{rotating}
\usepackage{xr} %
\usepackage[titles]{tocloft}
\usepackage{pdfpages}

\newcommand{\MONE}{$\f{MON}_=$\xspace}

\title{The Boolean Solution Problem\\ from the Perspective of Predicate Logic\\ --
  Extended Version --}

\author{Christoph Wernhard}
\institute{University of Potsdam, Germany}

\newcounter{equivcounter}
\newcounter{entailcounter}

\newcommand{\arity}[1]{\f{arity}(#1)}

\newcommand{\deq}{\Leftrightarrow}
\newcommand{\ndeq}{\not \Leftrightarrow}
\newcommand{\dimp}{\Rightarrow}
\newcommand{\drevimp}{\Leftarrow}

\newcommand{\tup}[1]{\boldsymbol{#1}}
\newcommand{\pps}{\tup{p}}
\newcommand{\qqs}{\tup{q}}

\newcommand{\uus}{\tup{u}}

\renewcommand{\xxs}{\tup{x}}
\newcommand{\XX}{\mathcal{X}}
\newcommand{\yys}{\tup{y}}

\newcommand{\Gs}{\tup{G}}
\newcommand{\Hs}{\tup{H}}
\newcommand{\Ts}{\tup{T}}
\newcommand{\Fs}{\tup{F}}

\newcommand{\Rs}{\tup{R}}
\newcommand{\Ss}{\tup{S}}
\renewcommand{\tts}{\tup{t}}

\newcommand{\ccs}{\tup{c}}

\newcommand{\SSS}{\tup{S}}

\newcommand{\SP}{SP\xspace}
\newcommand{\SPs}{SPs\xspace}

\newcommand{\UNSP}{\mbox{1-SP}\xspace}
\newcommand{\UNSPs}{\mbox{1-SPs}\xspace}

\newcommand{\PSP}{PSP\xspace}

\newcommand{\RSP}{RSP\xspace}

\newcommand{\UNRSP}{\mbox{1-RSP}\xspace}
\newcommand{\UNRSPs}{\mbox{1-RSPs}\xspace}

\newcommand{\emptyseq}{\epsilon}

\renewcommand{\mimp}{\;\Rrightarrow\;}

\renewcommand{\mequi}{\;\Lleftarrow\hspace{-0.75em}\Rrightarrow\;}

\renewcommand{\mimp}{\;\mathit{implies}\;}

\renewcommand{\mequi}{\;\mathit{iff}\;}

\def\squareforxed{$\vartriangleleft$\hspace*{-0.01em}}

\def\xed{\ifmmode\squareforqed\else{\unskip\nobreak\hfil
\penalty50\hskip1em\null\nobreak\hfil\squareforxed
\parfillskip=0pt\finalhyphendemerits=0\endgraf}\fi}

\newcommand{\FC}{\mathcal{F}}
\newcommand{\FCI}{\mathcal{F}}

\newcommand{\PC}{\mathcal{P}}

\newcommand{\tand}{\text{ and }}

\newcommand{\ssubst}[3]{\f{SUBST}(#1,#2,#3)}
\renewcommand{\subst}[3]{\warnSubstIsUndefined}

\newcommand{\clean}[1]{\f{CLEAN}(#1)}

\newcommand{\leftside}{\mathit{Left~side}}
\newcommand{\rightside}{\mathit{Right~side}}

\newcommand{\sol}[2]{#1{:}#2}

\newcommand{\free}[1]{\f{free}(#1)}

\newcommand{\msym}[1]{#1}

\newcommand{\PCOND}{PARA-CONDITION}
\renewcommand{\PCOND}{}
\newcommand{\PCONDX}{}

\newcommand{\subunit}[1]{\subsection{#1}}

\newcommand{\eref}[1]{Sect.~\ref{#1}}

\newcommand{\algoinput}{\noindent\textsc{Input:} }
\newcommand{\algomethod}{\noindent\textsc{Method:} }
\newcommand{\algooutput}{\noindent\textsc{Output:} }

\newcommand{\algoskip}{\vspace{1pt}}

\newcommand{\ELIM}{ELIM}

\newcommand{\WITA}{\mathit{ELIM\hyph WITNESS}}
\newcommand{\ALWIT}{\f{SOLVE\hyph BY\hyph WITNESSES}}
\newcommand{\USPA}{\mathit{1\hyph SOLVE}}
\newcommand{\ALSOLVEQ}{\f{SOLVE\hyph ON\hyph SECOND\hyph ORDER}}
\newcommand{\ALMSE}{\f{SOLVE\hyph SUCC\hyph ELIM}}

\begin{document}

\maketitle

\begin{abstract}
Finding solution values for unknowns in Boolean equations was a principal
reasoning mode in the \name{Algebra of Logic} of the 19th century.  Schröder
investigated it as \name{Auflösungsproblem} (\name{solution problem}). It is
closely related to the modern notion of Boolean unification. Today it is
commonly presented in an algebraic setting, but seems potentially useful also
in knowledge representation based on predicate logic.  We show that it can be
modeled on the basis of first-order logic extended by second-order
quantification.  A wealth of classical results transfers, foundations for
algorithms unfold, and connections with second-order quantifier elimination
and Craig interpolation become apparent.
Although for first-order inputs the set of solutions is recursively
enumerable, the development of constructive methods remains a challenge.  We
identify some cases that allow constructions, most of them based on Craig
interpolation.
\end{abstract}

\noindent
Revision: July 02, 2025 %

\newpage
\setcounter{tocdepth}{3}

\newcommand{\marktocheader}
{\markboth{\small \hfill Table of Contents}%
{\small Table of Contents \hfill}}

\cftsetindents{section}{0em}{1.5em}
\cftsetindents{subsection}{1.5em}{2.3em}
\renewcommand{\cftpartfont}{\bf\normalsize}
\renewcommand{\cftpartpagefont}{\bf\normalsize}
\renewcommand{\cftpartpresnum}{Part \marktocheader}
\renewcommand{\cftpartafterpnum}{\vspace{1.5ex}}
\makeatletter
\renewcommand{\@tocrmarg}{2.55em plus1fil}
\makeatother

\setlength{\cftbeforepartskip}{12pt}
\renewcommand{\cftpartafterpnum}{\vskip6pt}
\setlength{\cftbeforesecskip}{3pt}
\renewcommand{\cftsecafterpnum}{\vskip3pt}

\tableofcontents

\section{Introduction}
\label{sec-intro}

Finding solution values for unknowns in Boolean equations was a principal
reasoning mode in the \name{Algebra of Logic} of the 19th century. Schröder
\cite{schroeder:all} investigated it as \name{Auflösungsproblem}
(\name{solution problem}). It is closely related to the modern notion of
Boolean unification. For a given formula with occurrences of unknowns,
formulas are sought such that after substituting the unknowns with them the
given formula becomes valid or, dually, unsatisfiable. Of interest are also
most general solutions, condensed representations of all solution
substitutions. The \name{method of successive eliminations}, which traces back
to Boole, is a central technique here. Schröder investigated
\name{reproductive solutions} as most general solutions, anticipating the
concept of \name{most general unifier}.
A comprehensive modern formalization based on this material, along with
historic remarks, is presented by Rudeanu \cite{rudeanu:74} in the framework
of Boolean algebra. In automated reasoning, variations of these techniques
have been considered mainly in the late 80s and early 90s with the motivation
to enrich Prolog and constraint processing by Boolean unification with respect
to propositional formulas handled as terms
\cite{martin-nipkow:rings-early,buettner:simonis:87,martin-nipkow:rings-journal,martin:nipkov:boolean:89,kkr:90,kkr:95}.
An early implementation, based on \cite{rudeanu:74}, has been also described
in \cite{Sofronie89}. An implementation with BDDs of the algorithm from
\cite{buettner:simonis:87} is reported in \cite{carlsson:boolean:91}. The
$\mathrm{\Pi}^P_2$-completeness of Boolean unification with constants was
proven only later in \cite{kkr:90,kkr:95} and seemingly independently in
\cite{baader:boolean:98}. Schröder's results were developed further by
Löwenheim \cite{loewenheim:1908,loewenheim:1910}. A generalization of Boole's
method beyond propositional logic to relational monadic formulas
has been presented by
Behmann in the early 1950s
\cite{beh:50:aufloesungs:phil:1,beh:51:aufloesungs:phil:2}. Recently the
complexity of Boolean unification in a predicate logic setting has been
investigated for some formula classes, in particular for quantifier-free
first-order formulas \cite{eberhard}.
A brief discussion of Boolean reasoning in comparison with predicate logic can
be found in~\cite{brown:boolean:03}.

Here we remodel the solution problem formally along with basic classical
results and some new generalizations in the framework of first-order logic
extended by second-order quantification.  The main thesis of this work is that
it is possible and useful to apply second-order quantification consequently
throughout the formalization.  What otherwise would require meta-level
notation is then expressed just with formulas.  As will be shown, classical
results can be reproduced in this framework in a way such that applicability
beyond propositional logic, possible algorithmic variations, as well as
connections with second-order quantifier elimination and Craig interpolation
become visible.
Of course, methods to solve Boolean equations on first-order formulas do not
necessarily terminate.  However, the set of solutions is recursively
enumerable. By the modeling in predicate logic we try to pin down the
essential points of divergence from propositional logic.
Special cases that allow solution construction are identified, most of them
related to definiens computation by Craig interpolation.

The envisaged application scenario is to let solving ``solution problems'', or
Boolean equation solving, on the basis of predicate logic join reasoning modes
such as second-order quantifier elimination (or ``semantic forgetting''), Craig
interpolation and abduction to support the mechanized reasoning about
relationships between theories and the extraction or synthesis of subtheories
with given properties. On the practical side, the aim is to relate it to
reasoning techniques such as Craig interpolation on the basis of first-order
provers, SAT and QBF solving, and second-order quantifier elimination based on
resolution \cite{scan} and the Ackermann approach \cite{dls}.
Numerous applications of Boolean equation solving in various fields are
summarized in \cite[Chap.~14]{rudeanu:01}.  Applications in automated theorem
proving and proof compression are mentioned in \cite[Sect.~7]{eberhard}.  The
prevention of certain redundancies has been described as application of
(concept) unification in description logics \cite{baader:narendran:01}.  Here
the synthesis of definitional equivalences is sketched as an application.

The rest of the paper is structured as follows: Notation, in particular for
substitution in formulas, is introduced in Sect.~\ref{sec-notation}.
In~Sect.~\ref{sec-solprob} a formalization of the solution problem is
presented and related to different points of view. Section~\ref{sec-mse} is
concerned with abstract properties of and algorithmic approaches to solution
problems with several unknowns. Conditions under which solutions exist are
discussed in Sect.~\ref{sec-existence}. Adaptations of classical material on
reproductive solutions are given in Sect.~\ref{sec-reproductive}. In
Sect.~\ref{sec-construction} various techniques for solution construction in
particular cases are discussed. Section~\ref{sec-conclusion} closes the paper
with concluding remarks.

The material in Sect.~\ref{sec-notation}--\ref{sec-existence} has also been
published as \cite{cw:boolean:frocos}.

\section{Notation and Preliminaries}
\label{sec-notation}

\subunit{Notational Conventions}

We consider formulas in first-order logic with equality, extended by
second-order quantification upon predicates.
They are constructed from atoms (including equality atoms), constant operators
$\true$, $\false$, the unary operator $\lnot$, binary operators $\land, \lor$
and quantifiers $\forall, \exists$ with their usual meaning.  Further binary
operators~$\imp, \revimp, \equi$, as well as $n$-ary versions of $\land$ and
$\lor$ can be understood as meta-level notation.  The operators $\land$ and
$\lor$ bind stronger than $\imp$, $\revimp$ and $\equi$. The scope of $\lnot$,
the quantifiers, and the $n$-ary connectives is the immediate subformula to
the right.
A subformula occurrence has in a given formula \defname{positive (negative)
  polarity} if it is in the scope of an even (odd) number of negations.

A \defname{vocabulary} is a set of \defname{symbols}, that is, predicate
symbols (briefly \defname{predicates}), function symbols (briefly
\defname{functions}) and \defname{individual symbols}. (Individual symbols are
not partitioned into variables and constants.  Thus, an individual symbol is
-- like a predicate -- considered as variable if and only if it is bound by a
quantifier.) The arity of a predicate or function $s$ is denoted by
$\arity{s}$.
The set of symbols that occur \defname{free} in a formula~$F$ is denoted by
$\free{F}$.  The property that no member of $\free{F}$ is bound by a
quantifier occurrence in $F$ is expressed as $\clean{F}$.  Symbols not present
in the formulas and other items under discussion are called \defname{fresh}.
The \defname{clean variant of} a given formula~$F$ is the formula~$G$ obtained
from~$F$ by successively replacing all bound symbols with fresh symbols such
that $\clean{G}$. The replacement is done in some way not specified here
further such that each formula has a \emph{unique} clean variant.

We write $F \entails G$ for \name{$F$ entails $G$}; $\valid F$ for \name{$F$
  is valid}; and $F \equiv G$ for \name{$F$ is equivalent to $G$}, that is,
\name{$F \entails G$ and $G \entails F$}.

We write \defname{sequences} of symbols, of terms and of formulas by
juxtaposition.  Their length is assumed to be finite. The empty sequence is
written $\emptyseq$. A sequence with length $1$ is not distinguished from its
sole member.  In contexts where a set is expected, a sequence stands for the
set of its members.
Atoms are written in the form $p(\tts)$, where $\tts$ is a sequence of terms
whose length is the arity of the predicate~$p$.  Atoms of the form
$p(\emptyseq)$, that is, with a nullary predicate $p$, are written also as
$p$.  For a sequence of \defname{fresh} symbols we assume that its members are
distinct.  A sequence $p_1\ldots p_n$ of predicates is said to \defname{match}
another sequence $q_1 \ldots q_m$ if and only if $n=m$ and for all $i \in
\{1,\ldots,n\}$ it holds that $\arity{p_i} = \arity{q_i}$.  If $\sss =
s_1\ldots s_n$ is a sequence of symbols, then $\forall \sss$ stands for
$\forall s_1 \ldots \forall s_n$ and $\exists \sss$ for $\exists s_1 \ldots
\exists s_n$.  

If $\Fs = F_1\ldots F_n$ is a sequence of formulas, then $\clean{\Fs}$ states
$\clean{F_i}$ for all $i \in \{1,\ldots,n\}$, and $\free{\Fs} =
\bigcup_{i=1}^n\free{F_i}$.
If $\Gs = G_1\ldots G_n$ is a second sequence of
formulas, then $\Fs \equiv \Gs$ stands for $F_1 \equiv G_1$ \name{and \ldots
  and} $F_n \equiv G_n$.

As explained below, in certain contexts the individual symbols in the set $\XX
= \{x_i \mid i \geq 1\}$ play a special role. For example, in the following
shorthands for a predicate~$p$, a formula~$F$ and $\xxs = x_1\ldots
x_{\arity{p}}$: $p \deq F$ stands for $\forall \xs\, (p(\xxs) \equi F)$; $p
\ndeq F$ for $\lnot (p \deq F)$; $p \dimp F$ for $\forall \xs\, (p(\xxs) \imp
F)$; and $p \drevimp F$ for $\forall \xs\, (p(\xxs) \revimp F)$.

\subunit{Substitution with Terms and Formulas}
\label{sec-not-form-subst}
To express systematic substitution of individual symbols and predicates
concisely, we use the following notation.

\begin{itemize}

\item \textit{$F(\ccs)$ and $F(\tts)$ 
-- Notational Context for Substitution of Individual Symbols.}  Let $\ccs =
c_1\ldots c_n$ be a sequence of distinct individual symbols.  We write $F$ as
$F(\ccs)$ to declare that for a sequence~$\tts = t_1 \ldots t_n$ of terms the
expression $F(\tts)$ denotes $F$ with, for $i \in \{1,\ldots,n\}$, all free
occurrences of $c_i$ replaced by $t_i$.
\item \textit{$F[\pps]$, $F[\Gs]$ and $F[\qqs]$ --
Notational Context for Substitution of Predicates.}  Let $\pps = p_1\ldots
p_n$ be a sequence of distinct predicates and let $F$ be a formula.  We write
$F$ as $F[\pps]$ to declare the following.
\begin{itemize}

\item \label{item-nota-pg} For a sequence~$\Gs = G_1(x_1 \ldots
  x_{\arity{p_1}}) \ldots G_n(x_1 \ldots x_{\arity{p_n}})$ of formulas the
  expression $F[\Gs]$ denotes $F$ with, for $i \in \{1,\ldots, n\}$, each atom
  occurrence $p_i(t_1 \ldots t_{\arity{p_i}})$ where $p_i$ is free in $F$
  replaced by $G_i(t_1 \ldots t_{\arity{p_i}})$.

\item \label{item-nota-pq} For a sequence $\qqs = q_1 \ldots q_n$ of
  predicates that matches $\pps$
  the expression $F[\qqs]$ denotes $F$ with, for $i \in \{1,\ldots, n\}$, each
  free occurrence of $p_i$ replaced by $q_i$.

\item The above notation $F[\Ss]$, where $\Ss$ is a sequence of formulas or of
  predicates, is generalized to allow also $p_i$ at the $i$th position of
  $\Ss$, for example $F[G_1 \ldots G_{i-1} p_{i} \ldots p_n]$. The formula
  $F[\Ss]$ then denotes $F$ with only those predicates $p_i$ with $i \in
  \{1,\ldots,n\}$ that are not present at the $i$th position in $\SSS$
  replaced by the $i$th component of $\SSS$ as described above (in the example
  only $p_{1},\ldots,p_{i-1}$ would be replaced).

\end{itemize}

\item \textit{$\Fs[\pps]$ -- Notational Context for Substitution in a Sequence
  of Formulas.} If $\Fs = F_1\ldots F_n$ is a sequence of formulas, then
  $\Fs[\pps]$ declares that $\Fs[\Ss]$, where $\Ss$ is a sequence with the
  same length as $\pps$, is to be understood as the sequence $F_1[\Ss] \ldots
  F_n[\Ss]$ with the meaning of the members as described above.
\end{itemize}
In the above notation for substitution of predicates by formulas the members
$x_1,\ldots, x_{\arity{p}}$ of $\XX$ play a special role: $F[\Gs]$ can be
alternatively considered as obtained by replacing predicates $p_i$ with
$\lambda$-expressions $\lambda x_1 \ldots \lambda x_{\arity{p_i}} . G_i$
followed by $\beta$-con\-version.  The shorthand $p \deq F$ can be
correspondingly considered as $p \equi \lambda x_1 \ldots \lambda
x_{\arity{p}} . G$.
The following property \name{substitutible} specifies
preconditions for meaningful simultaneous substitution of formulas for
predicates.
\begin{defn}[$\ssubst{\Gs}{\pps}{F}$ -- Substitutible Sequence of Formulas]
\label{def-subst-seq-new}
A sequence $\Gs = G_1\ldots G_m$ of formulas is called \defname{substitutible
  for} a sequence $\pps = p_1 \ldots p_n$ of distinct predicates \defname{in}
a formula $F$, written $\ssubst{\Gs}{\pps}{F}$, if and only if $m = n$ and for
all $i \in \{1,\ldots,n\}$ it holds that
\begin{enumerateinline}
\iteminline \label{enum-subst-1} No free occurrence of $p_i$ in $F$ is in the
scope of a quantifier occurrence that binds a member of $\free{G_i}$;
\iteminline \label{enum-subst-2} $\free{G_i} \cap \pps = \emptyset$; and
\iteminline $\free{G_i} \cap \{x_j \mid j > \arity{p_i}\} = \emptyset$.
\end{enumerateinline}
\end{defn}
The following propositions demonstrate the introduced notation for formula
substitution. It is well known that terms can be ``pulled out of'' and
``pushed in to'' atoms, justified by the equivalences $p(t_1\ldots t_n)\;
\equiv\; \exists x_1 \ldots \exists x_n\, (p(x_1 \ldots x_n) \land
\bigwedge_{i=1}^n x_i=t_i)\; \equiv\; \forall x_1 \ldots \forall x_n\, (p(x_1
\ldots x_n) \lor \bigvee_{i=1}^n x_i\neq t_i)$, which hold if no member of
$\{x_1,\ldots, x_n\}$ occurs in the terms $t_1,\ldots,t_n$. Analogously,
substitutible subformulas can be ``pulled out of'' and ``pushed in to''
formulas.
\begin{prop}[Pulling-Out and Pushing-In of Subformulas]
Let $\Gs = G_1\ldots G_n$ be a sequence of formulas, let $\pps = p_1\ldots
p_n$ be a sequence of distinct predicates and let $F = F[\pps]$ be a formula
such that $\ssubst{\Gs}{\pps}{F}$. Then

\smallskip

\slab{prop-pullout}
$F[\Gs]\; \equiv\; \exists \pps\, (F \land 
\bigwedge_{i=1}^{n} (p_i \deq G_i))\; \equiv\;
\forall \pps\, (F \lor \bigvee_{i=1}^{n} (p_i \not \deq G_i)).$

\slab{prop-fixing}
$\forall \pps\, F\; \entails\; F[\Gs]\; \entails\; \exists \pps\, F.$
\end{prop}
\name{Ackermann's Lemma} \cite{ackermann:35} can be applied in certain cases
to \defname{eliminate} second-order quantifiers, that is, to compute for a given
second-order formula an equivalent first-order formula.  It plays an important
role in many modern methods for elimination and semantic forgetting -- see,
e.g.,
\cite{dls,sqema,soqe,schmidt:2012:ackermann,ks:2013:frocos,ZhaoSchmidt15b}.
\begin{prop}[Ackermann's Lemma, Positive Version]
\label{prop-ackermann}
Let $F, G$ be formulas and let $p$ be a predicate such that
$\ssubst{G}{p}{F}$, $p \notin \free{G}$ and all free occurrences of $p$ in $F$
have negative polarity.  Then
$\exists p\, ((p \drevimp G) \land F[p])\; \equiv\; F[G]$.
\end{prop}

\section{The Solution Problem from Different Angles}
\label{sec-solprob}

\subsection{Basic Formal Modeling}

Our formal modeling of the Boolean solution problem is based on two concepts,
\name{solution problem} and \name{particular solution}.
\begin{defn}[{$F[\pps]$} -- Solution Problem (\SP), Unary Solution Problem (\UNSP)]
\label{def-sp}
A \defname{solution problem (\SP)} $F[\pps]$ is a pair of a formula $F$ and a
sequence~$\pps$ of distinct predicates.  The members of $\pps$ are called the
\name{unknowns} of the \SP. The length of $\pps$ is called the \defname{arity}
of the \SP.  A \SP with arity $1$ is also called \defname{unary solution
  problem (\UNSP)}.
\end{defn}
The notation $F[\pps]$ for solution problems establishes as a ``side effect''
a context for specifying substitutions of $\pps$ in $F$ by formulas as
specified in Sect.~\ref{sec-not-form-subst}.
\begin{defn}[Particular Solution]
\label{def-sol-particular} A \defname{particular solution} (briefly
\name{solution}) \defname{of} a \SP $F[\pps]$ is defined as a sequence
$\Gs$ of formulas such that
$\ssubst{\Gs}{\pps}{F}$ and \label{item-def-sol-particular-valid}
$\valid F[\Gs]$.
\end{defn}
The property $\ssubst{\Gs}{\pps}{F}$ in this definition implies that no member
of $\pps$ occurs free in a solution.
Of course, \name{particular solution} can also be defined on the basis of
unsatisfiability instead of validity, justified by the equivalence of $\valid
F[\Gs]$ and $\lnot F[\Gs] \entails \false$.  The variation based on validity has
been chosen here because then the associated second-order quantifications are
existential, matching the usual presentation of elimination techniques.

Being a \name{solution} is aside of the substitutibility condition a
\emph{semantic} property, that is, applying to formulas modulo equivalence: If
$\Gs$ is a solution of~$F[\pps]$, then all sequences $\Hs$ of formulas such
that $\ssubst{\Hs}{\pps}{F}$ and $\Gs \equiv \Hs$ (that is, if $\Gs =
G_1\ldots G_n$ and $\Hs = H_1\ldots H_n$, then $H_i \equiv G_i$ holds for all
$i \in \{1,\ldots,n\}$) are also solutions of $F[\pps]$.

\name{Solution problem} and \name{solution} as defined here provide
abstractions of computational problems in a technical sense that would be
suitable, e.g., for complexity analysis.  Problems in the latter sense can be
obtained by fixing involved formula and predicate classes.
The abstract notions are adequate to develop much of the material on the
``Boolean solution problem'' shown here.  
On occasion, however, we consider restrictions, in particular to propositional
and to first-order formulas, as well as to nullary predicates.  As shown in
\eref{sec-reproductive}, further variations of \name{solution}, general
representations of several particular solutions, can be introduced on the
basis of the notions defined here.

\begin{examp}[A Solution Problem and its Particular Solutions]
\label{examp-sp}
As an example of a solution problem consider $F[p_1p_2]$ where
\[
\begin{array}{rcl}
F & = & 
\forall \msym{x} \, (\msym{a}(\mathit{x}) \imp  \msym{b}(\mathit{x}))\; \imp \\
&& (\forall \msym{x} \, (\msym{p_{1}}(\mathit{x}) \imp  \msym{p_{2}}(\mathit{x})) \;\land\;
\forall \msym{x} \, (\msym{a}(\mathit{x}) \imp  \msym{p_{2}}(\mathit{x}))\;
 \land\;
\forall \msym{x} \, (\msym{p_{2}}(\mathit{x}) \imp  \msym{b}(\mathit{x}))).
\end{array}
\]
The intuition is that the antecedent $\forall \mathit{x}\,
(\msym{a}(\mathit{x}) \imp \msym{b}(\mathit{x}))$ specifies the ``background
theory'', and w.r.t. that theory the unknown $p_1$ is ``stronger'' than the
other unknown $p_2$, which, in addition, is ``between'' $a$ and $b$.  Examples
of solutions are: $a(x_1)a(x_1)$; $a(x_1)b(x_1)$; $\false a(x_1)$;
$b(x_1)b(x_1)$; and $(a(x_1) \land b(x_1))(a(x_1) \lor b(x_1))$. No solutions
are, for example: $b(x_1)a(x_1)$; $a(x_1)\false$; and all members of $\{\true,
\false\} \times \{\true, \false\}$.
\end{examp}
Assuming a countable vocabulary, the set of valid first-order formulas is
recursively enumerable. It follows that for an $n$-ary \SP $F[\pps]$ where $F$
is first-order the set of those of its particular solutions that are sequences
of first-order formulas is also recursively enumerable: An $n$-ary sequence
$\Gs$ of well-formed first-order formulas that satisfies the syntactic
restriction $\ssubst{\Gs}{\pps}{F}$ is a solution of $F[\pps]$ if and only if
$F[\Gs]$ is valid.

In the following subsections further views on the solution problem will be
discussed: as unification or equation solving, as a special case of
second-order quantifier elimination, and as related to determining
definientia and interpolants.

\subsection{View as Unification}
\label{sec-unification}

Because $\valid F[\Gs]$ if and only if $F[\Gs] \equiv \true$, a particular
solution of $F[\pps]$ can be seen as a unifier of the two formulas $F[\pps]$
and $\true$ modulo logical equivalence as equational theory. From the
perspective of unification, the two formulas appear as terms, the members of
$\pps$ play the role of variables, and the other predicates play the role of
constants. 

Vice versa, a unifier of two formulas can be seen as a particular solution,
justified by the equivalence of $L[\Gs] \equiv R[\Gs]$ and $\valid (L \equi
R)[\Gs]$, which holds for sequences $\Gs$ and $\pps$ of formulas and
predicates, respectively, and formulas $L = L[\pps], R = R[\pps]$, $(L \equi
R) = (L \equi R)[\pps]$ such that $\ssubst{\Gs}{\pps}{L}$ and
$\ssubst{\Gs}{\pps}{R}$.
This view of formula unification can be generalized to sets with a finite
cardinality $k$ of equivalences, since \name{for all $i \in \{1,\ldots,k\}$ it
  holds that $L_i \equiv R_i$} can be expressed as $\valid \bigwedge_{i =
  1}^{k}(L_i \equi R_i)$.

An exact correspondence between solving a solution problem $F[p_1\ldots p_n]$
where $F$ is a propositional formula with $\lor,\land,\lnot,\false,\true$ as
logic operators and \mbox{E-uni}\-fi\-cation with constants in the theory of
Boolean algebra (with the mentioned logic operators as signature) applied to
$F =_E \true$ can be established: Unknowns $p_1,\ldots,p_n$ correspond to
variables and propositional atoms in $F$ correspond to constants.  A
particular solution $G_1\ldots G_n$ corresponds to a unifier $\{p_1 \leftarrow
G_1, \ldots, p_n \leftarrow G_n\}$ that is a ground substitution.  The
restriction to ground substitutions is due to the requirement that unknowns do
not occur in solutions.
General solutions (see \eref{sec-reproductive}) are expressed with further special
parameter atoms, different from the unknowns. These correspond to fresh
variables in unifiers.

A generalization of Boolean unification to predicate logic with various
specific problems characterized by the involved formula classes has been
investigated in \cite{eberhard}. The material presented here 
is largely orthogonal to that work, but a technique from \cite{eberhard} has
been adapted to more general cases in \eref{sec-ehw}.

\subsection{View as Construction of Elimination Witnesses}

Another view on the solution problem is related to eliminating second-order
quantifiers by replacing the quantified predicates with ``witness formulas''.
\begin{defn}[\ELIM-Witness]
\label{def-so-witness}
Let $\pps = p_1 \ldots p_n$ be a sequence of distinct predicates.  An
\defname{\ELIM-witness of} $\pps$ \defname{in} a formula $\exists \pps\,
F[\pps]$ is defined as a sequence $\Gs$ of formulas such that
$\ssubst{\Gs}{\pps}{F}$ and $\exists \pps\, F[\pps] \equiv F[\Gs]$.
\end{defn}
The condition $\exists \pps\, F[\pps] \equiv F[\Gs]$ in this definition is
equivalent to $\valid \lnot F[\pps] \lor F[\Gs]$.
If $F[\pps]$ and the considered $\Gs$ are first-order, then finding an
\ELIM-witness is second-order quantifier elimination on a first-order argument
formula, restricted by the condition that the result is of the form $F[\Gs]$.
Differently from the general case of second-order quantifier elimination on
first-order arguments, the set of formulas for which elimination succeeds
and, for a given formula, the set of its elimination results, are then
recursively enumerable.
Some well-known elimination methods yield \ELIM-witnesses, for example rewriting
a formula that matches the left side of Ackermann's Lemma
(Prop.~\ref{prop-ackermann}) with its right side, which becomes evident when
considering that the right side $F[G]$ is equivalent to $\forall x_1\ldots
\forall x_{\arity{p}}\, (G \revimp G) \land F[G]$.
Finding particular solutions and finding \ELIM-witnesses can be expressed in
terms of each other.
\begin{prop}[Solutions and \ELIM-Witnesses]
\label{prop-wit-sol}
Let $F[\pps]$ be a \SP and let $\Gs$ be a sequence of formulas. Then

\slab{prop-wit-ito-sol} $\Gs$ is an \ELIM-witness of $\pps$ in $\exists
\pps\, F$ if and only if $\Gs$ is a solution of the \SP $(\lnot F[\qqs] \lor
F)[\pps]$, where $\qqs$ is a sequence of fresh predicates matching $\pps$.

\slab{prop-sol-ito-wit} $\Gs$ is a solution of $F[\pps]$ if and only if
$\Gs$ is an \ELIM-witness  of $\pps$ in 
$\exists \pps\, F$ and it holds that
$\valid \exists \pps\, F$.
\end{prop}
\begin{proof}[Sketch]%
Assume $\ssubst{\Gs}{\pps}{F}$.  (\ref{prop-wit-ito-sol}) Follows since
$\exists \pps\, F[\pps] \equiv F[\Gs]$ iff $\exists \pps\, F[\pps]
\entails F[\Gs]$ iff $F[\pps] \entails F[\Gs]$ iff $\valid \lnot
F[\qqs] \lor F[\Gs]$.
(\ref{prop-sol-ito-wit}) Left-To-Right: Follows since $\valid F[\Gs]$ implies
$\valid \exists \pps\, F[\pps]$ and $\valid F[\Gs]$, which implies $\exists
\pps\, F[\pps] \equiv \true \equiv F[\Gs]$.  Right-to-left: Follows since
$\exists \pps\, F[\pps] \equiv F[\Gs]$ and $\valid \exists \pps\, F[\pps]$
together imply $\valid F[\Gs]$.  \qed
\end{proof}

\subsection{View as Related to Definientia and Interpolants}
\label{view-definientia}

The following proposition shows a further view on the solution problem that
relates it to definitions of the unknown predicates.
\begin{prop}[Solution as Entailed by a Definition]
\label{prop-char-sol-def}
A sequence $\Gs = G_1\ldots G_n$ of formulas is a particular solution of a \SP
$F[\pps = p_1\ldots p_n]$ if and only if $\ssubst{\Gs}{\pps}{F}$ and
$\bigwedge_{i=1}^n (p_i \deq G_i) \entails F$.
\end{prop}

\begin{proof}%
Follows from the definition of \name{particular solution} and
Prop.~\ref{prop-pullout}.
\qed
\end{proof}
In the special case where $F[p]$ is a \UNSP with a \emph{nullary}
unknown $p$, the
characterization of a solution~$G$ according to Prop.~\ref{prop-char-sol-def}
can be expressed with an entailment where a definition of the unknown~$p$
appears on the right instead of the left side: If $p$ is nullary, then $\lnot
(p \deq G) \equiv p \deq \lnot G$. Thus, the statement $p \deq G \entails F$
is for nullary~$p$ equivalent to
\begin{equation}
\label{eq-char-definiens}
\lnot F \entails p \deq \lnot G.
\end{equation}
The second condition of the characterization of \name{solution} according to
Prop.~\ref{prop-char-sol-def}, that is, $\ssubst{G}{p}{F}$, holds if it is
assumed that $p$ is not in $\free{G}$, that $\free{G} \subseteq \free{F}$ and
that no member of $\free{F}$ is bound by a quantifier occurrence in $F$. A
solution is then characterized as negated definiens of $p$ in the negation of
$F$.
Another way to express (\ref{eq-char-definiens}) along with the condition that
$G$ is semantically independent from $p$ is as follows.
\begin{equation}
\exists p\, (\lnot F \land \lnot p)\; \entails\;
G\; \entails\; \lnot \exists p\, (\lnot F \land p).
\end{equation}
The second-order quantifiers upon the nullary~$p$ can be
eliminated, yielding the following equivalent statement.
\begin{equation}
\label{eq-char-ipol}
\lnot F[\false]\; \entails\; G\; \entails\; F[\true].
\end{equation}
Solutions~$G$ then appear as the formulas in a range, between $\lnot
F[\false]$ and $F[\true]$. This view is reflected in
\cite[Thm.~2.2]{rudeanu:74}, which goes back to work by Schröder. If $F$ is
first-order, then Craig interpolation can be applied to compute formulas~$G$
that also meet the requirements $\free{G} \subseteq \free{F}$ and $p \notin
\free{F}$ to ensure $\ssubst{G}{p}{F}$.  Further connections to Craig
interpolation are discussed in \eref{sec-construction}.

\section{The Method of Successive Eliminations -- Abstracted}
\label{sec-mse}

\subsection{Reducing $n$-ary to $1$-ary Solution Problems}

The \name{method of successive eliminations} to solve an $n$-ary solution
problem by reducing it to unary solution problems is attributed to Boole and
has been formally described in a modern algebraic setting in \cite[Chapter~2,
  \S~4]{rudeanu:74}. It has been rediscovered in the context of Boolean
unification in the late 1980s, notably with~\cite{buettner:simonis:87}.
Rudeanu notes in \cite[p.~72]{rudeanu:74} that variations described by several
authors in the 19th century are discussed by Schröder
\cite[vol.~1, \S\S~26,27]{schroeder:all}.  
To research and compare all variations up to now seems to be a major undertaking
on its own. Our aim is here to provide a foundation to derive and analyze
related methods.  The following proposition formally states the core property
underlying the method in a way that, compared to the Boolean algebra version
in \cite[Chapter~2, \S~4]{rudeanu:74}, is more abstract in several aspects:
Second-order quantification upon predicates that represent unknowns plays the
role of meta-level shorthands that encode expansions; no commitment to a
particular formula class is made, thus the proposition applies to second-order
formulas with first-order and propositional formulas as special cases; it is
not specified how solutions of the arising unary solution problems are
constructed; and it is not specified how intermediate second-order formulas
(that occur also for inputs without second-order quantifiers) are handled. The
algorithm descriptions in the following subsections show different
possibilities to instantiate these abstracted aspects.
\begin{prop}[Characterization of Solution Underlying the Method of
Successive Eliminations] 
\label{prop-mse}
Let $F[\pps = p_1\ldots p_n]$ be a \SP and let $\Gs = G_1\ldots G_n$ be a
sequence of formulas. Then the following statements are \mbox{equivalent.}

\begin{enumequi}
\item \label{prop-char-mse-sol} $\Gs$ is a solution of $F[\pps]$.
\item \label{prop-char-mse-inward}
For $i \in \{1, \ldots, n\}$:
$G_i$ is a solution of the \UNSP \[(\exists p_{i+1} \ldots \exists p_n\,
F[G_1\ldots G_{i-1} p_i \ldots p_n])[p_i]\]
such that $\free{G_i} \cap \pps = \emptyset$.
\end{enumequi}
\end{prop}
\begin{proof}%
Left-to-right: From~\ref{prop-char-mse-sol} it follows that $\valid
F[\Gs]$. Hence, for all $i \in \{1,\ldots,n\}$ by Prop.~\ref{prop-fixing} it
follows that \[\valid \exists p_{i+1} \ldots \exists p_n\, F[G_1\ldots G_i
  p_{i+1} \ldots p_n].\] From~\ref{prop-char-mse-sol}
it also follows
that $\ssubst{\Gs}{\pps}{F}$. This implies that for all $i \in \{1,\ldots,n\}$
it holds that \[\ssubst{G_i}{p_i}{\exists p_{i+1} \ldots \exists p_n\,
  F[G_1\ldots G_{i-1} p_{i} \ldots p_n]} \text{ and } \free{G_i} \cap \pps =
\emptyset.\] We thus have derived for all $i \in \{1,\ldots,n\}$ the two
properties that characterize $G_i$ as a solution of the \UNSP as stated 
in~\ref{prop-char-mse-inward}.

Right-to-left: From~\ref{prop-char-mse-inward} it follows that $G_n$ is a
solution of the \UNSP \[(F[G_1\ldots G_{n-1} p_n])[p_n].\] Hence, by the
characteristics of \name{solution} it follows that $\valid F[G_1\ldots
  G_{n}]$.  The property $\ssubst{\Gs}{\pps}{F}$ can be derived from
$\free{\Gs} \cap \pps = \emptyset$ and the fact that for all $i \in
\{1,\ldots,n\}$ it holds that $\ssubst{G_i}{p_i}{(\exists p_{i+1} \ldots
  \exists p_n\, F[G_1\ldots G_{i-1} p_i \ldots p_n])}$.
The properties $\valid F[G_1\ldots G_{n}]$ and $\ssubst{\Gs}{\pps}{F}$
characterize $\Gs$ as a solution of the \SP $F[\pps]$.  \qed
\end{proof}

\noindent
This proposition states an equivalence between the solutions of an $n$-ary \SP
and the solutions of $n$ \UNSPs.  These \UNSPs are on formulas with an
existential second-order prefix. The following gives an example of this
decomposition.
\begin{examp}[Reducing an $n$-ary Solution Problem to Unary Solution Problems]
Consider the \SP $F[p_1p_2]$ of Examp.~\ref{examp-sp}.  The \UNSP with
unknown~$p_1$ according to Prop.~\ref{prop-mse} is \[(\exists p_2\,
F[p_1p_2])[p_1],\] whose formula is, by second-order quantifier elimination,
equivalent to
$\forall \mathit{x} \, (\msym{a}(\mathit{x}) \imp
\msym{b}(\mathit{x})) \imp \forall \mathit{x}\, (\msym{p_{1}}(\mathit{x}) 
\imp \msym{b}(\mathit{x}))$.
Take $\msym{a}(x_1)$ as solution $G_1$ of that \UNSP.  The \UNSP with unknown
$p_2$ according to Prop.~\ref{prop-mse} is
 \[(F[G_1p_2])[p_2].\] Its formula is then, by replacing 
in formula $F$ from Examp.~\ref{examp-sp} predicate $\msym{p_{1}}$
with~$\msym{a}$, followed by removing the resulting duplicate conjunct,
equivalent to
\[
\forall \msym{x} \, (\msym{a}(\mathit{x}) \imp  \msym{b}(\mathit{x}))\; \imp\;
 (\forall \msym{x} \, (\msym{a}(\mathit{x}) \imp  \msym{p_{2}}(\mathit{x}))\;
 \land\;
\forall \msym{x} \, (\msym{p_{2}}(\mathit{x}) \imp  \msym{b}(\mathit{x}))).
\]
A solution of that second \UNSP is, for example, $\msym{b}(\mathit{x_1})$,
yielding the pair $\msym{a}(\mathit{x_1})\msym{b}(\mathit{x_1})$ as solution
of the originally considered \SP $F[p_1p_2]$.
\end{examp}

\subsection{Solving on the Basis of Second-Order Formulas}

The following algorithm to compute particular solutions is an immediate
transfer of Prop.~\ref{prop-mse}. Actually, it is more an ``algorithm
template'', since it is parameterized with a method to compute \UNSPs and
covers a nondeterministic as well as a deterministic variation.
\begin{algo}[$\ALSOLVEQ$]
\label{algo-solveq}
Let $\FC$ be a class of formulas and let $\USPA$ be a nondeterministic or a
deterministic algorithm that outputs for \UNSPs of the form $(\exists
p_1\ldots \exists p_n\,F[p])[p]$ with $F \in \FC$ solutions $G$ such that
$\free{G} \cap \{p_1,\ldots,p_n\} = \emptyset$ and $F[G] \in \FC$.

\algoskip

\algoinput A \SP $F[p_1\ldots p_n]$, where $F \in \FC$, that has a
solution.

\algoskip

\algomethod
For $i := 1$ to $n$ do: Assign to $G_i$ an output of $\USPA$ applied to the
\UNSP
$(\exists p_{i+1} \ldots \exists p_n\,
F[G_1\ldots G_{i-1} p_i \ldots p_n])[p_i].$

\algoskip

\algooutput The sequence $G_1\ldots G_n$ of formulas, which is a
particular solution of $F[p_1\ldots p_n]$.
\end{algo}
The solution components $G_i$ are successively assigned to some solution of
the \UNSP given in Prop.~\ref{prop-mse}, on the basis of the previously
assigned components $G_1\ldots G_{i-1}$.  Even if the formula $F$ of the input
problem does not involve second-order quantification, these \UNSPs are on
second-order formulas with an existential prefix $\exists p_{i+1} \ldots
\exists p_n$ upon the yet ``unprocessed'' unknowns.

The algorithm comes in a nondeterministic and a deterministic variation, just
depending on whether $\USPA$ is instantiated by a nondeterministic or a
deterministic algorithm. Thus, in the nondeterministic variation the
nondeterminism of $\USPA$ is the only source of nondeterminism.  With
Prop.~\ref{prop-mse} it can be verified that if a nondeterministic $\USPA$ is
``complete'' in the sense that for each solution there is an execution path
that leads to the output of that solution, then also $\ALSOLVEQ$ based on it
enjoys that property, with respect to the $n$-ary solutions $G_1\ldots G_n$.

For the deterministic variation, from Prop.~\ref{prop-mse} it follows that if
$\USPA$ is ``complete'' in the sense that it outputs some solution whenever a
solution exists, then, given that $F[p_1\ldots p_n]$ has a solution, which is
ensured by the specification of the input, also $\ALSOLVEQ$ outputs some
solution $G_1\ldots G_n$.

This method applies $\USPA$ to existential second-order formulas, which
prompts some issues for future research: As indicated in
Sect.~\ref{view-definientia} (and elaborated in \eref{sec-construction}) Craig
interpolation can in certain cases be applied to compute solutions of \UNSPs.
Can QBF solvers, perhaps those that encode QBF into predicate logic
\cite{qbf:pl}, be utilized to compute Craig interpolants?  Can it be useful to
allow second-order quantifiers in solution formulas because they make these
smaller and can be passed between different calls to $\USPA$?

As shown in \eref{sec-reproductive}, if $\USPA$ is a method that outputs
so-called \name{reproductive} solutions, that is, most general solutions that
represent all particular solutions, then also $\ALSOLVEQ$ outputs reproductive
solutions. Thus, there are two ways to obtain representations of all
particular solutions whose comparison might be potentially interesting: A
deterministic method that outputs a single reproductive solution and the
nondeterministic method with an execution path to each particular solution.

\subsection{Solving with the Method of Successive Eliminations}

The \name{method of successive eliminations} in a narrower sense is applied in
a Boolean algebra setting that corresponds to propositional logic and outputs
reproductive solutions. The consideration of reproductive solutions belongs to
the classical material on Boolean reasoning
\cite{schroeder:all,loewenheim:1910,rudeanu:74} and is modeled in the present
framework in \eref{sec-reproductive}.
Compared to $\ALSOLVEQ$, the method handles the second-order quantification by
eliminating quantifiers one-by-one, inside-out, with a specific method and
applies a specific method to solve \UNSPs, which actually
yields reproductive solutions. These incorporated methods apply to
propositional input formulas (and to first-order input formulas if the
unknowns are nullary).  Second-order quantifiers are eliminated by rewriting
with the equivalence $\exists p\, F[p] \equiv F[\true] \lor F[\false]$. As
solution of a \UNSP $F[p]$ the formula $(\lnot F[\false] \land t) \lor
(F[\true] \land \lnot t)$ is taken,
where $t$ is a fresh nullary predicate that is considered specially. The
intuition is that particular solutions are obtained by replacing~$t$ with
arbitrary formulas in which $p$ does not occur (see \eref{sec-reproductive}
for a more in-depth discussion).

The following algorithm is an iterative presentation of the \name{method of
  successive eliminations}, also called \name{Boole's method}, in the variation
due to \cite{buettner:simonis:87}. The presentation in
\cite[Sect.~3.1]{martin:nipkov:boolean:89}, where apparently minor corrections
compared to \cite{buettner:simonis:87} have been made, has been taken here as
technical basis.  We stay in the validity-based setting, whereas
\cite{rudeanu:74,buettner:simonis:87,martin:nipkov:boolean:89} use the
unsatisfiability-based setting.  Also differently from
\cite{buettner:simonis:87,martin:nipkov:boolean:89} we do not make use of the
\name{xor} operator.
\begin{algo}[$\ALMSE$]

\algoinput A \SP $F[p_1\ldots p_n]$, where $F$ is propositional, that has a
solution and a sequence $t_1\ldots t_n$ of fresh nullary predicates.

\algoskip

\algomethod 
\begin{enumerate}
\item Initialize $F_n[p_1\ldots p_n]$ with $F$.
\item \label{algo-mse-iter-1} For $i := n$ to $1$ do: Assign to
  $F_{i-1}[p_1\ldots p_{i-1}]$ the formula $F_i[p_1\ldots p_{i-1}\true] \lor
  F_i[p_1\ldots p_{i-1}\false]$.

\item \label{algo-mse-iter-2}
For $i := 1$ to $n$ do: Assign to $G_i$ the formula $(\lnot F_i[G_1\ldots
  G_{i-1}\false] \land t_i) \lor (F_i[G_1\ldots G_{i-1}\true] \land \lnot
t_i).$
\end{enumerate}

\algoskip

\algooutput The sequence $G_1\ldots G_n$ of formulas, which is a
reproductive solution of $F[p_1\ldots p_n]$ with respect to
the special predicates $t_1\ldots t_n$.
\end{algo}
The formula assigned to $F_{i-1}$ in step~(\ref{algo-mse-iter-1}.)  is the
result of eliminating $\exists p_i$ in $\exists p_i\, F_{i}[p_1\ldots p_i]$
and the formula assigned to $G_i$ in step~(\ref{algo-mse-iter-2}.) is the
reproductive solution of the \UNSP $(F_i[G_1\ldots G_{i-1}p_i])[p_i]$,
obtained with the respective incorporated methods indicated above.
The recursion in the presentations of
\cite{buettner:simonis:87,martin:nipkov:boolean:89} is translated here into
two iterations that proceed in opposite directions: First, existential
quantifiers of $\exists p_1 \ldots \exists p_n\ F$ are eliminated inside-out
and the intermediate results, which do not involve second-order quantifiers,
are stored.  Solutions of \UNSPs are computed in the second phase on the basis
of the stored formulas.

In this presentation it is easy to identify two ``hooks'' where it is possible
to plug-in alternate methods that produce other outputs or apply to further
formula classes: In step~(\ref{algo-mse-iter-1}.)  the elimination method and
in step~(\ref{algo-mse-iter-2}.) the method to determine solutions of \UNSPs.
If the plugged-in method to compute \UNSPs outputs particular solutions, then
$\ALMSE$ computes particular instead of reproductive solutions.

\subsection{Solving by Inside-Out Witness Construction}

Like $\ALMSE$, the following algorithm eliminates second-order quantifiers
one-by-one, inside-out, avoiding intermediate formulas with existential
second-order prefixes of length greater than $1$, which arise with
$\f{SOLVE\hyph ON\hyph}$ $\f{SECOND\hyph ORDER}$.
In contrast to $\ALMSE$, it performs elimination by the computation of
\ELIM-witnesses.
\begin{algo}[$\ALWIT$] 
Let $\FC$ be a class of formulas and $\WITA$ be an algorithm that computes for
formulas $F \in \FC$ and predicates $p$ an \ELIM-witness~$G$ of $p$ in
$\exists p\, F[p]$ such that $F[G] \in \FC$.

\algoskip
\algoinput A \SP $F[p_1\ldots p_n]$, where $F \in \FC$, that has a solution.

\algoskip
\algomethod
For $i := n$ to $1$ do:
\begin{enumerate}
\item \label{step-algo-assign-gi} Assign to $G_i[p_1\ldots p_{i-1}]$ 
the output of $\WITA$ applied to
 \[\exists p_i\, F[p_1\ldots p_{i} G_{i+1}\ldots G_n].\]
\item \label{step-algo-reassign-gj} For $j := n$ to $i+1$ do: Re-assign to
  $G_j[p_1\ldots p_{i-1}]$ the formula $G_j[p_1\ldots p_{i-1}G_i]$.
\end{enumerate}

\algoskip
\algooutput
The sequence $G_1\ldots G_n$ of formulas, which
provides a particular solution of $F[p_1\ldots p_n]$.
\end{algo}
Step (\ref{step-algo-reassign-gj}.) in the algorithm expresses that a new
value is assigned to $G_j$ and that $G_j$ can be designated by $G_j[p_1\ldots
  p_{i-1}]$, justified because the new value does not contain free occurrences
of $p_{i}, \ldots, p_{n}$.  In step~(\ref{step-algo-assign-gi}.) the
respective current values of $G_{i+1}\ldots G_n$ are used to instantiate $F$.
It is not hard to see from the specification of the algorithm that for input
$F[\pps]$ and output $\Gs$ it holds that $\exists \pps\, F \equiv F[\Gs]$ and
that $\ssubst{\Gs}{\pps}{F}$. By Prop.~\ref{prop-sol-ito-wit}, $\Gs$ is then a
solution if $\valid \exists \pps\, F$. This holds indeed if $F[\pps]$ has a
solution, as shown below with Prop.~\ref{prop-exists-necessary}.

If $\WITA$ is ``complete'' in the sense that it computes an elimination
witness for all input formulas in $\FC$, then $\ALWIT$ outputs a solution.
Whether all solutions of the input \SP can be obtained as outputs for
different execution paths of a nondeterministic version of $\ALWIT$ obtained
through a nondeterministic $\WITA$, in analogy to the nondeterministic
variation of $\ALSOLVEQ$, appears to be an open problem.

\section{Existence of Solutions}
\label{sec-existence}

\subsection{Conditions for the Existence of Solutions}

We now turn to the question for the conditions under which there exists a
solution of a given \SP, or, in the terminology of \cite{rudeanu:74}, the \SP
is \name{consistent}. A necessary condition is easy to see.
\begin{prop}[Necessary Condition for the Existence of a Solution]
\label{prop-exists-necessary}
If a \SP $F[\pps]$ has a solution, then it holds that $\valid \exists \pps\, F$.
\end{prop}
\begin{proof}%
Follows from the definition of \name{particular solution} and
Prop.~\ref{prop-fixing}.  \qed
\end{proof}
Under certain presumptions that hold for
propositional logic this condition is also sufficient.  To express these
abstractly we use the following concept.

\begin{defn}[SOL-Witnessed Formula Class]
\label{def-witnessed-class}
A formula class $\FCI$ is called \defname{SOL-witnessed for} a predicate class
$\PC$ if and only if for all $p \in \PC$ and $F[p] \in \FCI$ the following
statements are equivalent.
\begin{enumequi}
\item $\valid \exists p\, F$.
\item \label{def-dwc-two} 
There exists a solution
$G$ of the \UNSP $F[p]$ such that $F[G] \in \FCI$.
\end{enumequi}
\end{defn}
Since the right-to-left direction of that equivalence holds in general, the
left-to-right direction alone would provide an alternate characterization.
The class of propositional formulas is SOL-witnessed (for the class of nullary
predicates). This follows since in propositional logic it holds
that 
\begin{equation}
\exists p\, F[p] \equiv F[F[\true]],
\end{equation} which can be derived in the following steps: 
$F[F[\true]] \;\equiv\; \exists
p\, (F[p] \land (p \equi F[\true])) \;\equiv\; (F[\true] \land (\true \equi
F[\true])) \lor (F[\false] \land (\false \equi F[\true])) \;\equiv\; F[\true]
\lor F[\false] \;\equiv\; \exists p\, F[p]$.

The following definition adds closedness under existential second-order
quantification and dropping of void second-order quantification to the notion
of \name{SOL-witnessed}, to allow the application on \UNSPs matching with
item~(b) in Prop.~\ref{prop-mse}.\hspace*{-1pt}%
\begin{defn}[MSE-SOL-Witnessed Formula Class]
\label{def-mwitnessed-class}
A formula class $\FCI$ is called \defname{MSE-SOL-witnessed} for a predicate
class $\PC$ if and only if it is SOL-witnessed for $\PC$, for all
$p \in \PC$ and $F \in \FCI$ it holds that $\exists p\, F \in
\FCI$, and, if $\exists p\, F \in \FCI$ and
$p \notin \free{F}$, then $F \in \FCI$.
\end{defn}
The class of existential QBFs (formulas of the form $\exists \pps\, F$ where
$F$ is propositional) is MSE-SOL-witnessed (like the more general class of
QBFs -- second-order formulas with only nullary predicates). 
Another example is the class of first-order formulas extended by second-order
quantification upon nullary predicates, which is MSE-SOL-witnessed for the
class of nullary predicates. The following proposition can be seen as
expressing an invariant of the method of successive eliminations that holds
for formulas in an MSE-SOL-witnessed class.
\begin{prop}[Solution Existence Lemma]
\label{prop-sol-existence-lemma}
Let $\FCI$ be a formula class that is MSE-SOL-witnessed for predicate class
$\PC$. Let $F[\pps = p_1\ldots p_n] \in \FCI$ with $\pps \in \PC^n$. If
$\valid \exists \pps\, F[\pps]$, then for all $i \in \{0,\ldots, n\}$ there
exists a sequence $G_1\ldots G_{i}$ of formulas such that $\free{G_1\ldots
  G_{i}} \cap \pps = \emptyset$, $\ssubst{G_1\ldots G_{i}}{p_1\ldots
  p_{i}}{F}$, $\valid \exists p_{i+1}\ldots \exists p_n F[G_1\ldots
  G_{i}p_{i+1}\ldots p_n]$ and \mbox{$\exists p_{i+1}\ldots \exists p_n
  F[G_1\ldots G_{i}p_{i+1}\ldots p_n]\! \in\! \FCI$.}
\end{prop}
\begin{proof}%
By induction on the length $i$ of the sequence $G_1\ldots G_i$.  
The
conclusion of the proposition holds for the base case $i=0$: The statement
$\ssubst{\emptyseq}{\emptyseq}{F}$ holds trivially, $\valid \exists \pps\, F$
is given as precondition, and $\exists \pps\, F \in \FCI$ follows from $F \in
\FCI$.
For the induction step, assume that the conclusion of the proposition holds
for some $i \in \{0, \ldots, n-1\}$. That is, there exists a sequence
$G_1\ldots G_{i}$ of formulas such that $\free{G_1\ldots G_{i}} \cap \pps =
\emptyset$, $\ssubst{G_1\ldots G_{i}}{p_1\ldots p_{i}}{F}$, $\valid \exists
p_{i+1} \ldots \exists p_n\, F[G_1\ldots G_{i}p_{i+1}\ldots p_n]$ and $\exists
p_{i+1} \ldots \exists p_n\, F[G_1\ldots G_{i}p_{i+1}\ldots p_n] \in \FCI$.
Since $\FCI$ is MSE-SOL-witnessed for $\PC$ and $p_{1}, \ldots, p_{i} \in \PC$
it follows that there exists a solution $G_{i+1}$ of the \UNSP 
\[(\exists p_1
\ldots \exists p_i \exists p_{i+2} \ldots \exists p_n \, F[G_1\ldots
  G_{i}p_{i+1}\ldots p_n])[p_{i+1}]\] such that
$\exists p_1 \ldots \exists p_i \exists p_{i+2} \ldots \exists p_n\, 
F[G_1\ldots G_{i+1}p_{i+2}\ldots p_n] \in \FCI$.  
From the characteristics of \name{solution} it follows that
\[\ssubst{G_{i+1}}{p_{i+1}}{\exists p_1 \ldots \exists p_i \exists p_{i+2}
  \ldots \exists p_n \, F[G_1\ldots G_{i}p_{i+1}\ldots p_n]},\] which implies
(since all members of $\pps$ with exception of $p_{i+1}$ are in the quantifier
prefix of the problem formula) that $\free{G_{i+1}} \cap \pps = \emptyset$,
hence \[\free{G_1\ldots G_{i+1}} \cap \pps = \emptyset.\] Given the induction
hypothesis $\ssubst{G_1\ldots G_{i}}{p_1\ldots p_{i}}{F}$, it also implies
\[\ssubst{G_1\ldots G_{i+1}}{p_1\ldots p_{i+1}}{F}.\]
From the characteristics of \name{solution} it follows in addition that
\[\valid \exists p_1 \ldots \exists p_i \exists p_{i+2} \ldots \exists p_n\,
F[G_1\ldots G_{i+1}p_{i+2}\ldots p_n],\] which, since
$\free{G_1\ldots G_{i+1}} \cap \pps = \emptyset$, is equivalent to
\[\valid \exists p_{i+2} \ldots \exists p_n\,
F[G_1\ldots G_{i+1}p_{i+2}\ldots p_n].\] Finally, we conclude from $\exists
p_1 \ldots \exists p_i \exists p_{i+2} \ldots \exists p_n\, F[G_1\ldots
  G_{i+1}p_{i+2}\ldots p_n] \in \FCI$, established above, and the definition
of \name{MSE-SOL-witnessed} that \[\exists p_{i+2} \ldots \exists p_n\,
F[G_1\ldots G_{i+1}p_{i+2}\ldots p_n] \in \FCI,\] which completes the proof of
the induction step.
\qed
\end{proof}
A sufficient and necessary condition for the existence of a solution of
formulas in MSE-SOL-witnessed classes now follows from
Prop.~\ref{prop-sol-existence-lemma} and Prop.~\ref{prop-exists-necessary}.
\begin{prop}[Existence of a Solution]
\label{prop-sol-ex}
Let $\FCI$ be a formula class that is MSE-SOL-witnessed on predicate class
$\PC$.  Then for all $F[\pps] \in \FCI$ where the members of $\pps$ are in
$\PC$ the following statements are equivalent.
\begin{enumequi}
\item $\valid \exists \pps\, F$.
\item %
There exists a solution
$\Gs$ of the \SP $F[\pps]$ such that $F[\Gs] \in \FCI$.
\end{enumequi}
\end{prop}
\begin{proof}%
Follows from Prop.~\ref{prop-sol-existence-lemma} and
Prop.~\ref{prop-exists-necessary}.
\qed
\end{proof}
From that proposition it is easy to see that for \SPs with propositional
formulas the complexity of determining the existence of a solution is the same
as the complexity of deciding validity of existential QBFs, as proven in
\cite{kkr:90,kkr:95,baader:boolean:98}, that is,
$\mathrm{\Pi}^P_2$-completeness: By Prop.~\ref{prop-sol-ex}, a \SP $F[\pps]$
where $F$ is propositional has a solution if and only if the existential QBF
$\exists \pps\, F[\pps]$ is valid and, vice versa, an arbitrary existential
QBF $\exists \pps\, F[\pps]$ (where $F$ is quantifier-free) is valid if and
only if the \SP $F[\pps]$ has a solution.

\subsection{Characterization of SOL-Witnessed in Terms of \ELIM-Witness}

The following proposition shows that under a minor syntactic precondition on
formula classes, \name{SOL-witnessed} can also be characterized in terms of
\name{\ELIM-witness} instead of \name{solution} as in
Def.~\ref{def-witnessed-class}.
\begin{prop}[SOL-Witnessed in Terms of \ELIM-Witness]
\label{prop-witness-val-equi}
Let $\FCI$ be a class of formulas that satisfies the following properties: For
all $F[p] \in \FCI$ and predicates $q$ with the same arity of $p$ it holds
that $F[p] \lor \lnot F[q] \in \FCI$, and for all $F \lor G \in \FCI$ it holds
that $F \in \FCI$.  The class $\FCI$ is SOL-witnessed for a predicate class
$\PC$ if and only if for all $p \in \PC$ and $F[p] \in \FCI$ there exists an
\ELIM-witness~$G$ of $p$ in $F[p]$ such that $F[G] \in \FCI$.
\end{prop}
\begin{proof}%
Left-to-right: Assume that $\FCI$ is meets the specified closedness conditions
and is SOL-witnessed for $\PC$, $p \in \PC$ and $F[p] \in \FCI$.
Let $q$ be a fresh predicate with the arity of $p$.  The obviously true
statement $\valid \exists p\, F[p] \lor \lnot \exists p\, F[p]$ is equivalent
to $\valid \exists p\, F[p] \lor \lnot F[q]$ and thus to $\valid \exists p\,
(F[p] \lor \lnot F[q]).$ By the closedness properties of $\FCI$ it holds that
$F[p] \lor \lnot F[q] \in \FCI$. Since $\FCI$ is SOL-witnessed for $\PC$ it
thus follows from Def.~\ref{def-witnessed-class} that there exists a solution
$G$ of the \SP $(F[p] \lor \lnot F[q])[p]$ such that $(F[G] \lor \lnot F[q])
\in \FCI$, and, by the closedness properties, also $F[G] \in \FCI$.  From the
definition of \name{solution} it follows that $\valid F[G] \lor \lnot F[q]$,
which is equivalent to $\exists p\, F[p] \equiv F[G]$, and also that
$\ssubst{G}{p}{F[G] \lor \lnot F[q]}$, which implies $\ssubst{G}{p}{F[G]}$.
Thus $G$ is an SO-witness of $p$ in $F[p]$ such that $F[G] \in \FCI$.
Right-to-left: Easy to see from Prop.~\ref{prop-sol-ito-wit}.  \qed
\end{proof}

\subsection{The Elimination Result as Precondition of Solution Existence}

Proposition~\ref{prop-sol-ex} makes an interesting relationship between the
existence of a solution and second-order quantifier elimination apparent that
has been pointed out by Schröder \cite[vol.~1, \S~21]{schroeder:all} and
Behmann \cite{beh:50:aufloesungs:phil:1}, and is briefly reflected in
\cite[p.~62]{rudeanu:74}: The formula $\exists \pps\, F$ is valid if and only
if the result of eliminating the existential second-order prefix (called
\name{Resultante} by Schröder \cite[vol.~1, \S~21]{schroeder:all}) is
valid. If it is not valid, then, by Prop.~\ref{prop-sol-ex}, the \SP $F[\pps]$
has no solution, however, in that case the elimination result represents the
\emph{unique (modulo equivalence) weakest precondition under which the \SP
  would have a solution.}  The following proposition shows a way to make this
precise.
\begin{prop}[The Elimination Result is the Unique Weakest Precondition of Solution\
Existence] 
\label{prop-wps}
Let $\FC$ be a formula class and let $\PC$ be a predicate class
  such that $\FC$ is MSE-SOL-witnessed on $\PC$. Let $F[\pps]$ be a solution
  problem where $F \in \FC$ and all members of $\pps$ are in $\PC$.  Let $A$
  be a formula such that $(A \imp F) \in \FC$, $A \equiv \exists \pps\, F$, and
  no member of $\pps$ occurs in $A$.  Then

\slab{prop-wps-one} The \SP $(A \imp F)[\pps]$ has a solution.

\slab{prop-wps-mid} If $B$ is a formula such that $(B \imp F) \in \FC$, no
member of $\pps$ occurs in $B$, and the \SP $(B \imp F)[\pps]$ has a solution,
then $B \entails A$.
\end{prop}

\begin{proof}
(\ref{prop-wps-one}) From the specification of $A$ it follows that $\valid A
  \imp \exists \pps F$ and thus $\valid \exists \pps\, (A \imp F)$.  Hence, by
  Prop.~\ref{prop-sol-ex}, the \SP $(A \imp F)[\pps]$ has a solution.
(\ref{prop-wps-mid}) Let $B$ be a formula such that the left side of
  holds. With Prop.~\ref{prop-sol-ex} it follows that $\valid B \imp \exists
  \pps F$. Hence $B \entails \exists \pps F$.  Hence $B \entails A$.
\qed
\end{proof}
The following example illustrates Prop.~\ref{prop-wps}.
\begin{examp}[Elimination Result as Precondition for Solvability]
\label{examp-condition}
Consider the \SP $F[p_1p_2]$ where
\[
\begin{array}{rcl}
F & = & 
 \forall \msym{x} \, (\msym{p_{1}}(\mathit{x}) \imp  \msym{p_{2}}(\mathit{x})) \;\land\;
\forall \msym{x} \, (\msym{a}(\mathit{x}) \imp  \msym{p_{2}}(\mathit{x}))\;
 \land\;
\forall \msym{x} \, (\msym{p_{2}}(\mathit{x}) \imp  \msym{b}(\mathit{x})).
\end{array}
\]
Its formula is the consequent of the \SP considered in Examp.~\ref{examp-sp}.
Since $\exists p_1 \exists p_2\, F \equiv \forall \msym{x} \,
(\msym{a}(\mathit{x}) \imp \msym{b}(\mathit{x})) \not \equiv \true$, from
Prop.~\ref{prop-sol-ex} it follows that $F[p_1p_2]$ has no solution.  If,
however, the elimination result $\forall \msym{x} \, (\msym{a}(\mathit{x})
\imp \msym{b}(\mathit{x}))$ is added as an antecedent to~$F$, then the
resulting \SP, which is the \SP of Examp.~\ref{examp-sp}, has a solution.
\end{examp}

\section{Reproductive Solutions as Most General Solutions}
\label{sec-reproductive}

Traditionally, concise representations of \emph{all} particular solutions have
been central to investigations of the solution problem.  This section presents
adaptations of classic material to this end, due in particular to Schröder and
Löwenheim, and presented in a modern algebraic formalization by Rudeanu
\cite{rudeanu:74}.
The idea is that a general solution $\Gs[\tts]$ has parameter predicates
$\tts$ such that each instantiation $G[\Ts]$ with a sequence~$\Ts$ of formulas
is a particular solution and that for all particular solutions $\Hs$ there
exists a sequence $\Ts$ of formulas such that $\Hs \equiv \Gs[\Ts]$. In this
way, a general solution represents all solutions. A remaining difficulty is to
determine for a given particular solution $\Hs$ the associated $\Ts$. This is
remedied with so-called \name{reproductive solutions}, for which $\Hs$ itself
can be taken as~$\Ts$, that is, it holds that $\Gs[\Hs] \equiv \Hs$.

We give formal adaptations in the framework of predicate logic that center
around the notion of \name{reproductive solution}.  This includes precise
specifications of \name{reproductive solution} and two further auxiliary types
of solution.  A technique to construct a reproductive solution from a given
particular solution, known as Schröder's \name{rigorous solution} or
\name{Löwenheim's theorem} and a construction of reproductive solutions due to
Schröder, which succeeds on propositional formulas in general, is adapted.
Finally, a way to express reproductive solutions of $n$-ary \SPs in terms of
reproductive solutions of \UNSPs in the manner of the method of successive
eliminations is shown.

\subunit{Parametric, General and Reproductive Solutions}
\label{sec-repro}

The following definitions give adaptations of the notions of \name{parametric},
\name{general} and \name{reproductive solution} for predicate logic, based on
the modern algebraic notions in \cite{rudeanu:74,deschamps:72} as starting
point. 
\begin{defn}[Parametric and Reproductive Solution Problem (\PSP,
    \RSP, \UNRSP)]
\label{def-psp} 
A \defname{parametric solution problem (\PSP)} $\sol{F[\pps]}{\tts}$ is a pair
of a solution problem~$F[\pps]$ and a sequence $\tts$ of distinct predicates
such that $(\free{F} \cup \pps) \cap \tts = \emptyset$. The members of $\tts$
are called the \defname{solution parameters} of the \PSP.
If the sequences of predicates $\pps$ and $\tts$ are matching, then the \PSP
is called a \defname{reproductive solution problem (\RSP)}.  A \RSP with arity $1$ is also called \defname{unary
  reproductive solution problem (\UNRSP)}.
\end{defn}

\newcommand{\xskip}{\smallskip}

\begin{defn}[Parametric, General and Reproductive Solution]
\label{def-sol-general-various} Define the following notions.

\xskip

\sdlab{def-sol-parametric} A \defname{parametric solution} of a \PSP
$\sol{F[\pps]}{\tts}$ is a sequence $\Gs[\tts]$ of formulas such that
$\clean{\Gs}$, $\ssubst{\Gs}{\pps}{F}$ and for all sequences of formulas $\Hs$
such that $\ssubst{\Hs}{\tts}{\Gs}$ and $\ssubst{\Hs}{\pps}{F}$ it holds that
if there exists a sequence $\Ts$ of formulas such that
$\ssubst{\Ts}{\tts}{\Gs}$, $\ssubst{\Gs[\Ts]}{\pps}{F}$ and \[\Hs \equiv
\Gs[\Ts],\] then \[\valid F[\Hs].\]

\xskip

\sdlab{def-sol-general} A \defname{general solution} of a \PSP
$\sol{F[\pps]}{\tts}$ is a sequence $\Gs[\tts]$ of formulas such that the
characterization of \name{parametric solution} (Def.~\ref{def-sol-parametric})
applies, with the \name{if-then} implication supplemented by its converse.

\sdlab{def-sol-reproductive} A \defname{reproductive solution} of a \RSP
$\sol{F[\pps]}{\tts}$ is a sequence $\Gs[\tts]$ of formulas such that
\begin{enumerate}
\item \label{stmt-def-repro-para} $\Gs$ is a parametric solution of
  $\sol{F[\pps]}{\tts}$ and
\item \label{stmt-def-repro-main} For all sequences $\Hs$ of formulas such
  that $\ssubst{\Hs}{\tts}{\Gs}$ \PCONDX and\linebreak $\ssubst{\Hs}{\pps}{F}$ it holds
  that if \[\valid F[\Hs],\] then \[\Hs \equiv \Gs[\Hs].\]
\end{enumerate}
\end{defn}

\noindent
\name{Parametric solution} can be characterized more concisely than in
Def.~\ref{def-sol-parametric}, however, not showing the syntactic correspondence to
the characterization of \name{general solution} in Def.~\ref{def-sol-general}:

\begin{prop}[Compacted Characterization of Parametric Solution]
\label{prop-par-altchar} 
A parametric solution of a \PSP $\sol{F[\pps]}{\tts}$ is a sequence $G[\tts]$
of formulas such that $\clean{\Gs}$, $\ssubst{\Gs}{\pps}{F}$ and for all
sequences $\Ts$ of formulas such that $\ssubst{\Ts}{\tts}{\Gs}$,
$\ssubst{\Gs[\Ts]}{\pps}{F}$ it holds that \[\valid F[\Gs[\Ts]].\]
\end{prop}

\begin{proof}%
The left side of the proposition can be expressed as
\begin{center}
\begin{tabular}{L{5em}@{\hspace{1em}}R{2em}@{\hspace{1em}}L{15em}}
& (1) & $\clean{\Gs}$,\\
& (2) & $\ssubst{\Gs}{\pps}{F}$,\\
\multicolumn{3}{l}{and for all sequences $H$, $T$ of formulas it holds that}\\
if & (3) & $\ssubst{\Hs}{\tts}{\Gs}$, (III)\\
& (4) & $\ssubst{\Hs}{\pps}{F}$,\\
& (5) & $\ssubst{\Ts}{\tts}{\Gs}$,\\
& (6) & $\ssubst{\Gs[\Ts]}{\pps}{F}$ \tand\\
& (7) & $\Hs \equiv \Gs[\Ts]$,\\
then & (8) & $\valid F[\Hs]$.
\end{tabular}
\end{center}
The right side of the proposition can be expressed as
\begin{center}
\begin{tabular}{L{5em}@{\hspace{1em}}R{2em}@{\hspace{1em}}L{15em}}
& (9) & $\clean{\Gs}$,\\
& (10) & $\ssubst{\Gs}{\pps}{F}$,\\
\multicolumn{3}{l}{and for all sequences $T$ of formulas it holds that}\\
if & (11) & $\ssubst{\Ts}{\tts}{\Gs}$ \tand\\
& (12) & $\ssubst{\Gs[\Ts]}{\pps}{F}$,\\
then & (13) & $\valid F[\Gs[\Ts]]$.
\end{tabular}
\end{center}
Left-to-right: If $\Hs = \Gs[\Ts]$, then $\Hs \equiv \Gs[\Ts]$.  Thus, this
direction of the proposition follows if statements~(9)--(12) imply (1)--(6),
with $\Hs$ instantiated to $\Gs[\Ts]$.  Statements~(1), (2), (5) and~(6) are
(9), (10), (11) and~(12), respectively.  The instantiation of~(3), that is,
$\ssubst{\Gs[\Ts]}{\tts}{\Gs}$, follows from (10) and (11). The instantiation
of~(4) is $\ssubst{\Gs[\Ts]}{\pps}{F}$, which is, like~(6), identical to~(12).
Right-to-left: Statements~(1)--(7) imply (9)--(12). This holds since (1), (2),
(5) and~(6) are (9), (10), (11) and (12), respectively. Hence, assuming the
right side of the proposition, statements~(1)--(7) then imply (13), that is,
$\valid F[\Gs[\Ts]]$.  Statement~(13), (7) and~(6) imply~(8), that is
$\valid F[\Hs]$, which concludes the proof.
\qed
\end{proof}

\noindent
The essential relationships between particular, parametric, general and
reproductive solutions, as well as an alternate characterization of
\name{reproductive solution} implied by these, are gathered in the following
proposition.

\begin{prop}[Relationships Between the Solution Types]
\label{prop-sol-relships}
Let $\Gs = \Gs[\tts]$ be a sequence of formulas. Then

\xskip
\slab{prop-para-is-sol-both} $\Gs$ is a parametric solution of the \PSP
$\sol{F[\pps]}{\tts}$ if and only if $\clean{\Gs}$ and
$\Gs$ is a particular solution of the \SP $F[\pps]$.

\xskip

\slab{prop-par-all-sol} If $\Gs$ is a parametric solution of the \PSP
$\sol{F[\pps]}{\tts}$ and $\Ts$ is sequence of formulas such that
$\ssubst{\Ts}{\tts}{\Gs}$, $\ssubst{\Gs[\Ts]}{\pps}{F}$, then $\Gs[\Ts]$ is a
particular solution of the \SP $F[\pps]$.

\xskip 

\slab{prop-general-is-para} A general solution of a \PSP is also a parametric
solution of that \PSP.

\xskip

\slab{prop-particular-to-general} If $\Gs$ is a general solution of the \PSP
$\sol{F[\pps]}{\tts}$ and $\Hs$ is a particular solution of the \SP $F[\pps]$
such that \PCOND $\ssubst{\Hs}{\tts}{\Gs}$, then there exists a sequence~$\Ts$
of formulas such that $\ssubst{\Ts}{\tts}{\Gs}$, $\ssubst{\Gs[\Ts]}{\pps}{F}$
and \[\Hs \equiv \Gs[\Ts].\]

\xskip

\slab{prop-repro-is-general} A reproductive solution of a \RSP is also a
general solution of that \RSP.

\xskip

\slab{prop-par-halfrepro} If $\Gs$ is a parametric solution of the \RSP
$\sol{F[\pps]}{\tts}$, then for all sequences $\Hs$ of formulas such that
$\ssubst{\Hs}{\tts}{\Gs}$ \PCONDX and $\ssubst{\Hs}{\pps}{F}$ it holds that
if \[\Hs \equiv \Gs[\Hs],\] then \[\valid F[\Hs].\]

\slab{prop-repro-bidir} $\Gs$ is a reproductive solution of the \RSP
$\sol{F[\pps]}{\tts}$ if and only if
\begin{enumerate}
\item $\Gs$ is a parametric solution of $\sol{F[\pps]}{\tts}$ and
\item For all sequences $\Hs$ of formulas such that
$\ssubst{\Hs}{\tts}{\Gs}$ \PCONDX and\linebreak $\ssubst{\Hs}{\pps}{F}$ it holds that
\[\valid F[\Hs]\] if and only if \[\Hs \equiv \Gs[\Hs].\]
\end{enumerate}
\end{prop}
Before we come to the proof of Prop.~\ref{prop-sol-relships}, let us observe
that the conclusion of Prop.~\ref{prop-par-halfrepro} is
item~(\ref{stmt-def-repro-main}.) of the definiens of \name{reproductive
  solution} (Def.~\ref{def-sol-reproductive}) after replacing the
\name{if-then} implication there by its converse, and that
Prop.~\ref{prop-repro-bidir} characterizes \name{reproductive solution} like
its definition (Def.~\ref{def-sol-reproductive}), except that the definiens is
strengthened by turning the \name{if-then} implication in
item~(\ref{stmt-def-repro-main}.)  into an equivalence.

\begin{proof}[Proposition~\ref{prop-sol-relships}]
\ 

(\ref{prop-para-is-sol-both}) 
Left-to-right: Let $\qqs$ be a sequence of fresh predicates that matches
$\tts$ and assume that $\Gs[\tts]$ is a parametric solution of
$\sol{F[\pps]}{\tts}$.  Hence $\ssubst{\Gs}{\pps}{F}$ and $\valid
F[\Gs[\qqs]]$, which implies $\valid F[\Gs]$.  Thus $\Gs$ is a particular
solution of $F[\pps]$. Note that this direction of the proposition
requires the availability of fresh predicates in the vocabulary.
Right-to-left: Can be derived in the following steps explained below.
\[
\begin{array}{r@{\hspace{1em}}l@{\hspace{0.5em}}l}
(1) & \Gs[\tts] \text{ is a particular solution of } F[\pps].\\
(2) & \clean{\Gs}\\
(3) & \ssubst{\Gs}{\pps}{F}\\
(4) & \valid F[\Gs].\\
(5) & \ssubst{\Ts}{\tts}{\Gs}.\\
(6) & \ssubst{\Gs[\Ts]}{\pps}{F}.\\
(7) & \ssubst{\Ts}{\tts}{F[\Gs]}.\\
(8) & \valid \forall \tts\, F[\Gs].\\
(9) & \valid F[\Gs[\Ts]].\\
(10) & \Gs \text{ is a parametric solution of } \sol{F[\pps]}{\tts}.
\end{array}
\]
Step~(1) and~(2), where $\tts$ is some sequence of distinct predicates such
that $(\free{F} \cup \pps) \cap \tts = \emptyset$, form the left side of the
proposition.  Steps~(3) and (4) follow from (1) and the characteristics of
\name{particular solution}.  Let $\Ts$ be a sequence of formulas such that (5)
and (6) hold, conditions on the left side of Prop.~\ref{prop-par-altchar}.
Step~(7) follows from (5) and (6).
Step (8) follows from (4).  Step~(9) follows from~(7) and (8) by
Prop.~\ref{prop-fixing}. Finally, step~(10), the right side of the proposition,
follows from Prop.~\ref{prop-par-altchar} with (9), (2) and~(3).

\smallskip

(\ref{prop-par-all-sol}) The left side of the proposition includes
$\ssubst{\Gs[\Ts]}{\pps}{F}$ and, by Prop.~\ref{prop-par-altchar}, implies
$\valid F[\Gs[\Ts]]$, from which the right side follows.

\smallskip

(\ref{prop-general-is-para}) Immediate from the definition of \name{general
 solution} (Def.~\ref{def-sol-general}).

\smallskip

(\ref{prop-particular-to-general}) The left side of the proposition implies
$\ssubst{\Hs}{\tts}{\Gs}$,\linebreak $\ssubst{\Hs}{\pps}{F}$ and $\valid F[\Hs]$. The
right side then follows from the definition of \name{general solution}
(Def.~\ref{def-sol-general}).

\smallskip

(\ref{prop-repro-is-general}) By definition, a reproductive solution is also a
parametric solution.  Let $\Gs$ be a reproductive solution of
$\sol{F[\pps]}{\tts}$.  Let $\f{COND}$ stand for the following conjunction of
three statements:
\[\ssubst{\Hs}{\tts}{\Gs},\;
  \ssubst{\Hs}{\pps}{F}\; \mathit{and}\; \valid F[\Hs].\]
From the definition of \name{reproductive solution} it
immediately follows that for all sequences $\Hs$ of formulas such that
$\f{COND}$ it holds that $\Hs \equiv \Gs[\Hs]$.
From this it follows that for all sequences $\Hs$ of formulas such that
$\f{COND}$ it holds that $\ssubst{\Hs}{\tts}{\Gs}$,
$\ssubst{\Gs[\Hs]}{\pps}{F}$ and $\Hs \equiv \Gs[\Hs]$, which can be derived
as follows. The first of the statements on the right,
$\ssubst{\Hs}{\tts}{\Gs}$, is included directly in the left side, that is,
$\f{COND}$.  The second one, $\ssubst{\Gs[\Hs]}{\pps}{F}$, follows from
$\ssubst{\Hs}{\tts}{\Gs}$ and $\ssubst{\Hs}{\pps}{F}$ that are in $\f{COND}$
together with $\ssubst{\Gs}{\pps}{F}$, which holds since $\Gs$ is a parametric
solution.
The above implication also holds if $\Hs$ on its right side is replaced by a
supposedly existing~$\Ts$. It then forms the remaining requirement to show
that $\Gs$ is a general solution: For all sequences $\Hs$ of formulas such
that $\f{COND}$ there exists a sequence~$\Ts$ of formulas such that
$\ssubst{\Ts}{\tts}{\Gs}$, $\ssubst{\Gs[\Ts]}{\pps}{F}$ and $\Hs \equiv
\Gs[\Ts]$.

\smallskip

(\ref{prop-par-halfrepro}) 
Can be shown in the following steps, explained below.
\[
\begin{array}{r@{\hspace{1em}}l@{\hspace{0.5em}}l}
(1) & \ssubst{\Gs}{\pps}{F}.\\ 
(2) & \ssubst{\Hs}{\tts}{\Gs}.\\
(3) & \ssubst{\Hs}{\pps}{F}.\\
(4) & \Hs \equiv \Gs[\Hs].\\
(5) & \ssubst{\Gs[\Hs]}{\pps}{F}.\\  
(6) & F[\Gs[\Hs]] \entails \false.\\
(7) & F[\Hs] \entails \false.\\
\end{array}
\]
Assume that $\Gs$ is a parametric solution of the \RSP $\sol{F[\pps]}{\tts}$,
which implies~(1).  Let $\Hs$ be a sequence of formulas such that (2) and (3),
the preconditions of the converse of (as well as the unmodified)
item~(\ref{stmt-def-repro-main}.)  in the definition of \name{reproductive
  solution} (Def.~\ref{def-sol-reproductive}), hold.  Further assume (4), the
right side of item~(\ref{stmt-def-repro-main}.).  We prove the proposition by
deriving the left side of item~(\ref{stmt-def-repro-main}.).  Step~(5) follows
from~(1), (2) and~(3).  Step~(6) follows from~(2) and (5) by
Prop.~\ref{prop-par-altchar} since $\Gs$ is a parametric solution.  Finally,
step~(7), the left side of item~(\ref{stmt-def-repro-main}.), follows from (6)
and (4) with (3) and (5).

\smallskip

(\ref{prop-repro-bidir}) Follows from Prop.~\ref{prop-par-halfrepro}.
\qed
\end{proof}

Rudeanu \cite{rudeanu:74} notes that the concept of reproductive solution
seems to have been introduced by Schröder \cite{schroeder:all}, while the term
\name{reproductive} is due to Löwenheim \cite{loewenheim:1910}.  Schröder
calls the additional requirement that a reproductive solution must satisfy in
comparison with general solution \name{Adventivforderung} (\name{adventitious
  requirement}) and discusses it at length in \cite[vol.~3,
  \S~12]{schroeder:all}, describing it with \name{reproduzirt} %
\cite[vol.~3, p.~171]{schroeder:all}.

\subunit{The Rigorous Solution}

From any given particular solution~$\Gs$, a reproductive solution can be
constructed, called here, following Schröder's terminology \cite[vol.~3,
  \S~12]{schroeder:all}, the \name{rigorous solution associated with} $\Gs$.
In the framework of Boolean algebra, the analogous construction is
\cite[Theorem~2.11]{rudeanu:74}.

\begin{prop}[The Rigorous Solution]
\label{prop-sol-rigorous}
Let $\sol{F[\pps]}{\tts = t_1\ldots t_n}$ be a \RSP.  For $i \in \{1,\ldots,
n\}$ let $\xs_i$ stand for $x_1\ldots x_{\arity{t_i}}$.  Assume $\free{F} \cap
\XX = \emptyset$, $\ssubst{t_1(\xs_1)\ldots t_n(\xs_n)}{\pps}{F}$ and
$\ssubst{F[\tts]}{\pps}{F}$. If $\Gs = G_1 \ldots G_n$ is a particular
solution of that $\RSP$, then the sequence $\Rs = R_1\ldots R_n$ of formulas
defined as follows is a reproductive solution of that \RSP.
\[R_i \text{ is the clean variant of } (G_i(\xs) \land \lnot F[\tts]) \lor 
(t_i(\xs_i) \land F[\tts]).\]
\end{prop}
In the specification of $R_i$ the formula $G_i$ is written as $G_i(\xs)$ to
indicate that members of $\XX$ may occur there literally without being
replaced.  In the unsatisfiability-based setting, the $R_i$ would be
characterized as the clean variant of
\begin{equation}
(G_i(\xs) \land F[\tts]) \lor (t_i(\xs_i) \land \lnot
F[\tts]).
\end{equation}
 The proof of this proposition is based on the following lemma, a predicate
 logic analog to \cite[Lemma~2.3]{rudeanu:74} for the special case $n=1$,
 which is sufficient to prove Prop.~\ref{prop-sol-rigorous}: The effect of the
 lemma for arbitrary $n$ is achieved by an application of
 Prop.~\ref{prop-sol-rigorous} within an induction.
\begin{prop}[Subformula Distribution Lemma]
\label{prop-rigo-dist-lemma}
Let $p$ be a predicate (with arbitrary arity $\geq 0$), let $F[p]$ be a
formula and let $V, W, A$ be formulas such that $\ssubst{V}{p}{F}$,
$\ssubst{W}{p}{F}$,  $\ssubst{A}{p}{F}$ and, in addition, $\free{A} \cap \XX
 = \emptyset$.  It then holds that
\[F[(A \land V) \lor (\lnot A \land W)]\; \equiv\;
 (A \land F[V]) \lor (\lnot A \land F[W]).\]
\end{prop}
\begin{proof}%
Assume the preconditions of the proposition.  It follows that $\ssubst{(A
  \land V) \lor (\lnot A \land W)}{p}{F}$.  Making use of
Prop.~\ref{prop-pullout}, the conclusion of the proposition can be then be
shown in the following steps.
\[
\begin{array}{c@{\hspace{1em}}l}
& \leftside\\
\equiv & \exists p\, (F[p] \land (p \deq ((A \land V) \lor (\lnot A \land
W))))\\
\equiv &
(A \land \exists p\, (F[p] \land (p \deq ((A \land V) \lor (\lnot A \land
W)))))\; \lor\\
& (\lnot A \land \exists p\, 
(F[p] \land (p \deq ((A \land V) \lor (\lnot A \land W)))))\\
\equiv &
(A \land \exists p\, (F[p] \land (p \deq ((\true \land V) \lor (\false \land
W)))))\; \lor\\
& (\lnot A \land \exists p\, (F[p] \land (p \deq ((\false \land V) \lor (\true \land W)))))\\
\equiv & (A \land \exists p\, (F[p] \land (p \deq V)))
\lor (\lnot A \land \exists p\, (F[p] \land (p \deq W)))\\
\equiv & \rightside.
\end{array}
\vspace{-16pt}%
\]
\qed
\end{proof}
The preconditions in Prop.~\ref{prop-rigo-dist-lemma} permit that $x_1,\ldots,
x_{\arity{p}}$ may occur free in $V$ and $W$, whereas in $A$ no member of
$\XX$ is allowed to occur free. We are now ready to prove
Prop.~\ref{prop-sol-rigorous}.

\begin{proof}[Proposition~\ref{prop-sol-rigorous}]
By item~(\ref{stmt-def-repro-para}.) of the definition of \name{reproductive
  solution} (Def.~\ref{def-sol-reproductive}), $\Rs[\tts] = R_1[\tts]\ldots
R_n[\tts]$ is required to be a parametric solution for which by
Prop.~\ref{prop-par-altchar} three properties have to be shown. The first one,
$\clean{\Rs}$, is immediate since each member of $\Rs$ is the clean variant of
some formula.  The second one, $\ssubst{\Rs}{\pps}{F}$, is easy to derive from
the preconditions and the definition of $\Rs$.  The third one is an
implication that can be shown in the following steps, explained below.
\[
\begin{array}{r@{\hspace{1em}}l}
(1) & \ssubst{\Ts}{\tts}{\Rs}.\\
(2) & \ssubst{\Rs[\Ts]}{\pps}{F}.\\
(3) & \ssubst{\Gs}{\pps}{F}.\\
(4) & \valid F[\Gs].\\
(5) & \lnot F[\Ts] \land \lnot F[\Rs[\Ts]] \entails \lnot F[\Gs].\\
(6) & F[\Ts] \land \lnot F[\Rs[\Ts]] \entails \lnot F[\Ts].\\
(7) & F[\Gs] \entails F[\Rs[\Ts]].\\
(8) & \valid F[\Rs[\Ts]].\\
\end{array}
\]
Let $\Ts$ be a sequence of formulas such that statements (1) and (2), which
are on the left side of the implication to show, do hold.  We derive the right
side of the implication, that is $\valid F[\Rs[\Ts]]$.  Steps~(3) and (4) hold
since $\Gs$ is a particular solution.  Steps~(5) and~(6) can be shown by
induction based on the equivalences~(9) and~(10), respectively, below, which
hold for all $i \in \{0,\ldots,n-1\}$ and follow from
Prop.~\ref{prop-rigo-dist-lemma}.
\[
\begin{array}{r@{\hspace{1em}}c@{\hspace{1em}}l}
(9) &       & \lnot F[G_1 \ldots G_{i} R_{i+1}[\Ts] \ldots R_{n}[\Ts]]\\
& \equiv &
\lnot F[G_1 \ldots G_{i}\,((G_{i+1} \land \lnot F[\Ts]) \lor 
(T_{i+1} \land F[\Ts]))\, R_{i+2}[\Ts] \ldots R_{n}[\Ts]]\\
& \equiv &
(\lnot F[\Ts] \land \lnot F[G_1 \ldots G_{i+1}
                R_{i+2}[\Ts] \ldots R_{n}[\Ts]])\; \lor\\
&& (F[\Ts] \land 
\lnot F[G_1 \ldots G_{i}T_{i+1}
                R_{i+2}[\Ts] \ldots R_{n}[\Ts]]).\\[1ex]
(10)       && \lnot F[T_1 \ldots T_i R_{i+1}[\Ts] \ldots R_{n}[\Ts]]\\
& \equiv &
\lnot F[T_1 \ldots T_{i}\,((G_{i+1} \land \lnot F[\Ts]) \lor 
(T_{i+1} \land F[\Ts]))\, R_{i+2}[\Ts] \ldots R_{n}[\Ts]]\\
& \equiv &
(\lnot F[\Ts] \land \lnot F[T_1 \ldots T_i G_{i+1}
                R_{i+2}[\Ts] \ldots R_{n}[\Ts]])\; \lor\\
&& (F[\Ts] \land 
F[T_1 \ldots T_{i+1} R_{i+2}[\Ts] \ldots R_{n}[\Ts]]).
\end{array}
\]
The required preconditions of Prop.~\ref{prop-rigo-dist-lemma} are justified
there as follows, where $F^\prime$ stands for $F$ under the substitutions
indicated in (9) or (10), that is, the formula matched with the left side of
Prop.~\ref{prop-rigo-dist-lemma}.

\smallskip

\noindent
\begin{tabular}{c@{\hspace{1em}}l@{\hspace{1em}}l}
\textbf{--} & $\ssubst{G_{i+1}}{p_{i+1}}{\lnot F^\prime}$:& Follows from~(3).\\
\textbf{--} &  $\ssubst{T_{i+1}}{p_{i+1}}{\lnot F^\prime}$:&
 Follows from~(1) and~(2).\\
\textbf{--} & $\ssubst{F[\Ts]}{p_{i+1}}{\lnot F^\prime}$:& Follows from (1),
(2) and the precondition\\ && $\ssubst{F[\tts]}{\pps}{F}$.\\
\textbf{--} & $\free{F[\Ts]} \cap \XX = \emptyset$: & Follows from
(1), (2) and the precondition\\ && $\free{F}
\cap \XX = \emptyset$.
\end{tabular}

\smallskip

\noindent
Step~(7) follows from~(6) and~(5) and, finally, step~(8) follows from~(7)
and~(4).

Item~(\ref{stmt-def-repro-main}.) of the definition of \name{reproductive
  solution} follows since for all sequences of formulas $\Hs$ such that
$\ssubst{\Hs}{\tts}{\Rs}$ and $\ssubst{\Hs}{\pps}{F}$
(note that
$\ssubst{\Hs}{\tts}{\Gs}$ is implied by $\ssubst{\Hs}{\tts}{\Rs}$)
it holds that if $\valid F[\Hs]$, then
$\Hs \equiv \Rs[\Hs]$, or, equivalently, but more explicated,
it holds for all $i \in \{1,\ldots,n\}$ that
\[
\begin{array}{r@{\hspace{1em}}l}
& R_i[\Hs]\\
\equiv & (G_i[\Hs] \land \lnot F[\Hs]) \lor (H_i \land F[\Hs])\\
\equiv & (G_i[\Hs] \land \false) \lor (H_i \land \true)\\
\equiv & H_i.
\end{array}
\vspace{-16pt}%
\]
\qed
\end{proof}

The algebraic version \cite[Theorem~2.11]{rudeanu:74} is attributed there and
in most of the later literature to Löwenheim
\cite{loewenheim:1908,loewenheim:1910}, thus known as \name{Löwenheim's
  theorem for Boolean equations}. However, at least the construction for unary
problems appears to be in essence Schröder's \name{rigorose Lösung}
\cite[vol.~3, \S~12]{schroeder:all}.  (Löwenheim remarks in
\cite{loewenheim:1908} that the \name{rigorose Lösung} can be derived as a
special case of his theorem.)  Behmann comments that Schröder's discussion of
\name{rigorose Lösung} starts only in a late chapter of \name{Algebra der
  Logik}
mainly for the reason that only then suitable notation was available
\cite[Footnotes on p.~22f]{beh:50:aufloesungs:phil:1}.
Schröder \cite[vol.~3, p.~168]{schroeder:all} explains his term \name{rigoros}
as adaptation of \name{à la rigueur}, that is, \name{if need be}, because he
does not consider the rigorous solution as a satisfying representation of all
particular solutions. He notes that to detect all particular solutions on the
basis of the \name{rigorose Lösung}, one would have to test all possible
formulas $T$ as parameter value.
As remarked in \cite[p.~382]{martin-nipkow:rings-journal}, \name{Löwenheim's
  theorem} has been rediscovered many times, for example in
\cite{martin-nipkow:rings-early}.

\subunit{Schröder's Reproductive Interpolant} 
\label{sec-sipol}
For \UNRSPs of the
form 
\begin{equation}
\sol{((A \dimp p) \land (p \dimp B))[p]}{t},\end{equation}
the formula \begin{equation}A \lor (B \land
t(\xs)),\end{equation} where $\xs = x_1\ldots x_{\arity{p}}$, is a reproductive
solution.  This construction has been shown by Schröder and is also discussed
in \cite{beh:50:aufloesungs:phil:1}.  For the notion of \name{solution} based
on unsatisfiability instead of validity, the analogous construction applies to
\UNRSPs of the form \begin{equation}\sol{((A \land p) \lor (B \land \lnot
    p))[p]}{t}\end{equation} and yields
\begin{equation}B \lor (\lnot A \land t(\xs)).\end{equation} We call the solution
\defname{interpolant} because with the validity-based notion of
\name{solution} assumed here the unknown~$p$, and thus also the solution,
is ``between'' $A$ and $B$, that is, implied by $A$ and implying $B$.  The
following proposition makes the construction precise and shows its
justification.  The proposition is an adaptation of \cite[Lemma
  2.2]{rudeanu:74}, where \cite[vol.~1, \S~21]{schroeder:all} 
is given as source.
\begin{prop}[Schröder's Reproductive Interpolant]
\label{prop-schroeder-ipol-t1}
Let \[\sol{(F = \forall \yys\, (A(\yys) \imp p(\yys)) \land \forall \yys\,
  (p(\yys) \imp B(\yys)))[p]}{t},\] where $\yys$ is a sequence with the arity
of $p$ as length of distinct individual symbols not in $\XX$, be a \UNRSP that
has a solution.  Let $\xs = x_1\ldots x_{\arity{p}}$. Assume
$\ssubst{A(\xs)}{p}{F}$, $\ssubst{B(\xs)}{p}{F}$ and
$\ssubst{t(\xs)}{p}{F}$. Then the clean variant
of the following formula is a reproductive solution
of that \UNRSP:
\[A(\xs) \lor (B(\xs) \land t(\xs)).\]
\end{prop}
That $p$ does not occur free in $A$ or in $B$ is ensured by the
preconditions\linebreak
$\ssubst{A}{p}{F}$ and $\ssubst{B}{p}{F}$.  The symbols $\yys$
for the quantified variables indicate that these are independent from the
special meaning of the symbols in~$\XX$.

\begin{proof}[Proposition~\ref{prop-schroeder-ipol-t1}]
 Assume the preconditions of the proposition and let $G[t]$ stand for the
 clean variant of $A(\xs) \lor (B(\xs) \land t(\xs))$.  By
 item~(\ref{stmt-def-repro-para}.) of the definition of \name{reproductive
   solution} (Def.~\ref{def-sol-reproductive}), $G$ is required to be a
 parametric solution for which by Prop.~\ref{prop-par-altchar} three
 properties have to be shown. The first one, $\clean{G}$, is immediate since
 $G$ is a clean variant of some formula.
The second one, $\ssubst{G}{p}{F}$, easily follows from the
preconditions.  The third one is an implication that can be shown in the
following steps, explained below.
\[
\begin{array}{r@{\hspace{1em}}l}
(1) & \ssubst{A(\xs) \lor (B(\xs) \land T(\xs))}{p}{F}.\\
(2) & \valid \exists p\; 
(\forall \yys\, (A(\yys) \imp p(\yys)) \land \forall \yys\,
  (p(\yys) \imp B(\yys))).\\
(3) & \valid \forall \yys\, (A(\yys) \imp B(\yys)).\\
(4) & \valid \forall \yys\, (A(\yys) \imp (A(\yys) \lor (B(\yys) \land
  T(\yys))))\; \land\\
  & \hphantom{\valid} \forall \yys\, 
  ((A(\yys) \lor (B(\yys) \land T(\yys))) \imp B(\yys)).\\ 
 (4) & \valid F[A(\xxs) \lor (B(\xxs) \land T(\xxs))].\\
  \end{array}
\]
Let $T(\xs)$ be a formula such that statement (1), which is on the left side
of the implication to show, does hold.  We derive the right side of the
implication, that is, $\valid F[B(\xs) \lor (A(\xs) \land T(\xs))]$:
Step~(2) follows with Prop.~\ref{prop-exists-necessary} from the precondition
that the considered $\UNRSP$ has a solution.  Step~(3) follows from~(2) by
second-order quantifier elimination, for example with Ackermann's lemma
(Prop.~\ref{prop-ackermann}). The formulas to the right of $\valid$ in both
statements are equivalent.  Step~(4) follows from (3) by logic. Justified
by~(1), we can express~(4) as~(5), the right side of the implication to show.
Item~(\ref{stmt-def-repro-main}.) of the definition of \name{reproductive
  solution} follows since for all formulas~$H(\xs)$ such that
$\ssubst{H(\xs)}{t}{G}$, $\ssubst{H(\xs)}{p}{F}$, it holds that $\valid
F[H(\xs)]$ implies $H(\xs) \equiv G[H(\xs)]$, which can be derived in the
following steps.
\[\begin{array}{r@{\hspace{1em}}ll}
& \valid F[H(\xs)]\\
\mequi & \valid 
\forall \yys\, (A(\yys) \imp H(\yys)) \land \forall \yys\, 
  (H(\yys) \imp B(\yys))\\
\mequi & 
\valid \forall \yys\, (H(\yys) \equi (A(\yys) \lor H(\yys)))\; \tand\\
& \valid \forall \yys\, (H(\yys) \equi (B(\yys) \land H(\yys)))\\
\mimp &
H(\xs) \equiv A(\xs) \lor H(\xs)\; \tand\; H(\xs) \equiv B(\xs) \land H(\xs)\\
\mimp & H(\xs) \equiv A(\xs) \lor (B(\xs) \land H(\xs))\\
\mequi & H(\xs) \equiv G[H(\xs)].
\end{array}
\vspace{-16pt}%
\]
\qed
\end{proof}

As shown by Schröder, the (clean variants of the) following two formulas are
further reproductive solutions in the setting of
Prop.~\ref{prop-schroeder-ipol-t1}.
\begin{equation}
B(\xs) \land
(A(\xs) \lor t(\xs))\end{equation} and
\begin{equation}
\label{eq-sipol-atbt}
(A(\xs) \land \lnot t(\xs)) \lor
(B(\xs) \land t(\xs)).\end{equation} These two formulas and the solution
according to Prop.~\ref{prop-schroeder-ipol-t1} are all equivalent under the
assumption that a solution exists, that is, $\valid \exists p\, F$, which, by
second-order quantifier elimination, is equivalent to 
\begin{equation}
\valid \forall \yys\,
(A(\yys) \imp B(\yys)).
\end{equation}

Any \UNRSP~$\sol{F[p]}{t}$ where $F$ is a propositional formula or, more
generally, where the unknown $p$ is nullary, can be brought into the form
matching Prop.~\ref{prop-schroeder-ipol-t1} by systematically renaming bound
symbols and rewriting $F[p]$ with the equivalence
\begin{equation}
F[p] \equiv (\lnot F[\false] \imp p) \land (p \imp F[\true]).
\end{equation}
For the notion of \name{solution} based on unsatisfiability, the required form
can be obtained with the Shannon expansion
\begin{equation}
F[p] \equiv (F[\true] \land p) \lor (F[\false] \land \lnot p).
\end{equation}

\subunit{From Unary to $n$-ary Reproductive Solutions}
\label{sec-repro-mse}

If the solution of a \RSP is composed as suggested by Prop.~\ref{prop-mse}
from reproductive solutions of unary solution problems, then it is itself a
reproductive solution.

\begin{prop}[Composing a Reproductive Solution from Unary Reproductive Solutions]
\label{prop-rst-to-nary}
Let $\sol{F[\pps = p_1\ldots p_n]}{\tts = t_1\ldots t_n}$ be a \RSP.
If $\Gs[\tts] = G_1[\tts]\ldots G_n[\tts]$ is a sequence of formulas such that
for all $i \in \{1,\ldots,n\}$ it holds that $G_i$ is a reproductive solution
of the \UNRSP \[\sol{(\exists p_{i+1} \ldots \exists p_n\, F[G_1\ldots
    G_{i-1}p_i \ldots p_n])[p_i]}{t_i}\] and $\free{G_i} \cap (\pps \cup
t_{i+1}\ldots t_n) = \emptyset$, then $\Gs$ is a reproductive solution of the
considered \RSP $\sol{F[\pps]}{\tts}$.
\end{prop}

\begin{proof}%
Assume the preconditions and the left side of the proposition.  We show the
two items of the definition of \name{reproductive solution}
(Def.~\ref{def-sol-reproductive}) for $\Gs$.
Item~(\ref{stmt-def-repro-para}.), that is, $\Gs$ is a parametric solution of
$\sol{F[\pps]}{\tts}$, can be derived as follows. Each $G_i$, for $i \in
\{1,\dots,n\}$, is a reproductive solution of the associated \UNRSP. Hence, by
Prop.~\ref{prop-sol-relships} it is a general, hence parametric, hence
particular solution. By Prop.~\ref{prop-mse} it follows that $\Gs$ is a
particular solution of $F[\pps]$.  By Prop.~\ref{prop-para-is-sol-both} it is
then also a parametric solution of $\sol{F[\pps]}{\tts}$.
Item~(\ref{stmt-def-repro-main}.) of the definition of \name{reproductive
  solution} can be shown as follows.
First we note the following statement that was given as precondition:
\vspace{-2ex}
\begin{center}
\begin{tabular}{R{3cm}@{\hspace{1em}}L{12cm}}
(1) & For $i \in \{1,\ldots,n\}$ it holds that
$\free{G_i} \cap t_{i+1}\ldots t_n = \emptyset$.
\end{tabular}
\end{center}
For $i \in \{1,\ldots, n\}$ let
\[F_i[p_i\tts] \eqdef \exists p_{i+1} \ldots \exists p_n\, F[G_1[\tts]\ldots
  G_{i-1}[\tts]p_i \ldots p_n],\] that is, $F_i$ is the formula of the \UNSP
of which $G_i$ is a reproductive solution.  By the definition of
\name{reproductive solution} and the left side of the proposition it holds for
all formulas $H_i$ that if
\vspace{-2ex}
\begin{center}
\begin{tabular}{R{3cm}@{\hspace{1em}}L{12cm}}
(2) & $\ssubst{H_i}{t_i}{G_i}$\\
    & $\ssubst{H_i}{p_i}{F_i[p_i\tts]}, \tand$\\
    & $\valid F_i[H_i\tts],$\\
$\text{then}$ &\\
(3) & $H_i \equiv G_i[t_1\ldots t_{i-1}H_it_{i+1}\ldots t_n].$
\end{tabular}
\end{center}
From this and (1) it follows that all for all sequences of formulas $H_1
\ldots H_i$ it holds that if
\vspace{-2ex}
\begin{center}
\begin{tabular}{R{3cm}@{\hspace{1em}}L{12cm}}
(4) & $\ssubst{H_i}{t_i}{G_i}$\\
    & $\ssubst{H_i}{p_i}{F_i[p_{i}H_1\ldots H_{i-1} t_{i}\ldots t_n]}, \tand$\\
    & $\valid F_i[H_{i}H_1\ldots H_{i-1} t_{i}\ldots t_n],$\\
$\text{then}$ &\\
(5) & $H_i \equiv G_i[H_1\ldots H_it_{i+1}\ldots t_n].$
\end{tabular}
\end{center}
Now let $\Hs \equiv H_1\ldots H_n$ be a sequence of formulas such that
\vspace{-2ex}
\begin{center}
\begin{tabular}{R{3cm}@{\hspace{1em}}L{12cm}}
(6) & $\ssubst{\Hs}{\tts}{\Gs}$\\
    & $\ssubst{\Hs}{\pps}{F}, \tand$\\
    & $\valid F[\Hs].$
\end{tabular}
\end{center}
We prove item~(\ref{stmt-def-repro-main}) of
Def.~\ref{def-sol-reproductive} by showing $\Hs \equiv \Gs[\Hs]$, which is
equivalent to the statement that for all $i \in \{1,\ldots,n\}$ it holds that
$H_i \equiv G_i[\Hs]$, and, because of~(1), to the statement that for all $i
\in \{1,\ldots,n\}$ it holds that $H_i \equiv G_i[H_1\ldots H_i t_{i+1}\ldots
  t_n]$, which matches~(5). We thus can prove $\Hs \equiv \Gs[\Hs]$ by showing
that~(4), which implies~(5), holds for all $i \in \{1, \ldots, n\}$.  The
substitutivity conditions in~(4) follow from the substitutivity conditions
in~(6).  The remaining condition $\valid F_i[H_{i}H_1\ldots H_{i-1}
  t_{i}\ldots t_n]$ can be proven by induction.  As induction hypothesis
assume that for all $j \in \{1, \ldots, i-1\}$ it holds that $H_j \equiv
G_j[\Hs]$.  From $\valid F[\Hs]$ in~(6) it follows by Prop.~\ref{prop-fixing}
that $\valid \exists p_{i+1} \ldots \exists p_n F[H_1\ldots H_i p_{i+1} \ldots
  p_n]$.  With the induction hypothesis it follows that \[\valid \exists
p_{i+1} \ldots \exists p_n\, F[G_1[\Hs]\ldots G_{i-1}[\Hs]H_i p_{i+1}\ldots
  p_n],\] which, given the substitutivity conditions of~(6) and
$\ssubst{\Gs}{\pps}{F}$, which holds since $\Gs$ is a parametric solution,
can be expressed as \[\valid F_i[H_{i}H_1\ldots H_{i-1} t_{i}\ldots t_n],\]
such that all conditions of~(4) are satisfied and $H_i \equiv G_i[\Hs]$ can be
concluded.
\qed
\end{proof}

\noindent
This suggests to compute reproductive solutions of \emph{propositional}
formulas for a $n$-ary \SP by constructing Schröder's reproductive
interpolants for \UNSPs. Since for propositional formulas second-order
quantifier elimination succeeds in general, the construction of Schröder's
reproductive interpolant can then be performed on the basis of conventional
propositional formulas, without second-order quantifiers.

\section{Towards Constructive Solution Techniques}
\label{sec-construction}

On the basis of first-order logic, it seems that so far there is no general
constructive method for the computation of solutions. We discuss various
special cases where a construction is possible. Some of these relate to
applications of Craig interpolation. Recent work by Eberhard, Hetzl and Weller
\cite{eberhard} shows a constructive method for quantifier-free first-order
formulas. A generalization of their technique to relational monadic formulas
is shown, which, however, produces solutions that would be acceptable only
under a relaxed notion of substitutibility.

\subunit{Background: Craig Interpolation, Definability and Independence}
\label{sec-craig}
\label{sec-ipol-independence}

By Craig's interpolation theorem \cite{craig:linear}, if $F$ and $G$ are
first-order formulas such that $F \entails G$, then there exists an a
\defname{Craig interpolant} of $F$ and $G$, that is, a first-order
formula~$H$ such that \begin{equation}\free{H}
 \subseteq \free{F} \cap \free{G}\end{equation} and \begin{equation}F
\entails H \entails G.\end{equation}
\defname{Craig interpolants} can be constructed from proofs of $\valid F \imp
G$, as, for example, shown for tableaux in
\cite{smullyan:book:68,fitting:book}.  Lyndon's interpolation theorem
strengthens Craig's theorem by considering in addition that predicates in the
interpolant~$H$ occur only in polarities in which they occur in both side
formulas, $F$ and $G$.  In fact, practical methods for the construction of
interpolants from proofs typically compute such Craig-Lyndon interpolants.

One of the many applications of Craig interpolation is the construction of a
definiens for a given predicate: Let $F[p q_1\ldots q_k]$ be a first-order
formula such that $\free{F} \cap \XX = \emptyset$ and $p q_1\ldots q_k$ is a
sequence of distinct predicates and let $\xs$ stand for $x_1\ldots
x_{\arity{p}}$. Then $p$ \defname{is explicitly definable in terms of
  $(\free{F} \setminus \{q_1, \ldots, q_k\})$ within~$F$}, that is, there
exists a first-order formula $G$ such that \begin{equation}\free{G} \subseteq
  (\free{F} \setminus \{p, q_1, \ldots, q_k\}) \cup \xs\end{equation}
  and \begin{equation}F \entails p \deq G,\end{equation} if and only if $p$
  \defname{is implicitly definable in terms of $(\free{F} \setminus \{q_1,
    \ldots, q_k\})$ within~$F$}, that is,
\begin{equation}\label{eq-def-implicit}\exists p \exists q_1 \ldots \exists q_k\, (F \land p(\xs)) \entails \lnot
\exists p \exists q_1 \ldots \exists q_k\, (F \land \lnot
p(\xs)).\end{equation} Entailment~(\ref{eq-def-implicit}) holds if and only if the following
first-order formula is valid.
\begin{equation}F \land p(\xs) \imp \lnot (F[p^{\prime}q^\prime_1\ldots q^\prime_k]\land \lnot
p^\prime(\xs)),\end{equation} where $p^\prime q^\prime \ldots q^\prime_k$ is a
sequence of fresh predicates that matches $p q \ldots q_k$.  The
\defname{definientia}~$G$ of $p$ with the stated characteristics are exactly
the Craig interpolants of the two sides of that implication.  Substitutibility
$\ssubst{G}{p}{F}$ can be ensured by presupposing $\clean{F}$ and that no
members of $\XX$ are bound by a quantifier occurrence in $F$.

Another application of Craig interpolation concerns the independence of
formulas from given predicates: Second-order quantification allows to express
that a formula $F[\pps]$ is semantically independent from the set of the
predicates in $\pps$ as \begin{equation}\exists \pps\,F \equiv
  F,\end{equation} which is equivalent to $\exists \pps\, F \entails F$, and
thus, if $\qqs$ is a sequence of fresh predicates that matches $\pps$, also
equivalent to \begin{equation}\valid F[\qqs] \imp F.\end{equation} As observed
in \cite{otto:interpolation:2000}, any interpolant of $F[\qqs]$ and $F$ is
then equivalent to $F$ but its free symbols do not contain members of $\pps$,
that is, it is syntactically independent from $\pps$.  Thus, for a given
first-order formula semantic independence from a set of predicates can be
expressed as first-order validity and, if it holds, an equivalent formula that
is also syntactically independent can be constructed by Craig interpolation.
With Craig-Lyndon interpolation this technique can be generalized to take also
polarity into account, based on encoding of polarity sensitive independence as
shown here for negative polarity: That $F[p]$ is independent from predicate
$p$ in negative polarity but may well depend on $p$ in positive polarity can
be expressed as \begin{equation}\exists q\, (F[q] \land \forall \xs\, (q(\xs)
  \imp p(\xs))),\end{equation} where $\xs = x_1\ldots x_{\arity{p}}$ and $q$
is a fresh predicate with the same arity as $p$.

\subunit{Cases Related to Definability and Interpolation}
\label{sec-solcases}

The following list shows cases where for an $n$-ary \SP $F[\pps = p_1\ldots
  p_n]$ with first-order~$F$ and which has a solution a particular solution
can be constructed. Each of the properties that characterize these cases is
``semantic'' in the sense that if it holds for $F$, then it also holds for any
first-order formula equivalent to $F$. In addition, each property is at least
``semi-decidable'', that is, the set of first-order formulas with the property
is recursively enumerable. Actually, in the considered cases, for each
property a first-order formula can be constructed from $F$ that is valid if
and only if $F$ has the property.
For two of the listed cases, (\ref{solcase-def-nf}.) and
(\ref{solcase-def-nft}.),  the characterizing property implies the existence of
a solution.

\begin{enumerate}
\setlength{\parskip}{1ex}
\item \textit{Each unknown occurs free in $F$ only with a single polarity.}  A
  sequence of $\true$ and $\false$, depending on whether the respective
  unknown occurs positively or negatively, is then a solution.  That $F$ is
  semantically independent of unknowns in certain polarities, that is, is
  equivalent to a formula in which the unknowns do not occur in these
  polarities, can be expressed as first-order validity and a corresponding
  formula that is syntactically independent can be constructed by Craig-Lyndon
  interpolation.

\item \label{solcase-def-f} \textit{Each unknown is definable in the formula.}
  A sequence of definientia, which can be constructed with Craig
  interpolation, is then a solution. Rationale: Let $G_1\ldots G_n$ be
  definientia of $p_1\ldots p_n$, respectively, in $F$.  Under the
  assumption that there exists a solution $\Hs$ of $F[\pps]$ it holds that
  \[\true\; \entails\; F[\Hs]\; \entails\; \exists \pps\, F[\pps]\; \equiv\;
  \exists \pps\,
  (F[\pps] \land \bigwedge_{i=1}^n (p_i \deq G_i))\; \equiv\; F[\Gs].\]

\item \label{solcase-def-nf} \textit{Each unknown is definable in the negated
  formula.}  The sequence of negated definientia, which can be constructed
  with Craig interpolation, is then a solution.  Rationale: It holds in
  general that $p \deq G \entails p \ndeq \lnot G$. Hence, if $G_1\ldots G_n$
  are definientia of $p_1\ldots p_n$ , respectively, in $\lnot F$, then $\lnot
  F \entails \bigwedge_{i = 1}^{n} (p_i \deq G_i) \entails \bigvee_{i = 1}^{n}
  (p_i \deq G_i) \equiv \bigvee_{i = 1}^{n} (p_i \ndeq \lnot G_i)$. Thus
  \[\bigwedge_{i = 1}^{n} (p_i \deq \lnot G_i) \entails F,\] matching the
  characterization of \name{solution} in Prop.~\ref{prop-char-sol-def}.

\item \textit{Each unknown is nullary.} 
  \label{solcase-def-nullary}
  This specializes case~(\ref{solcase-def-nf}.): If a solution exists, then a
  nullary unknown is definable in the negated formula. For nullary
  predicates~$p$ it holds in general that \[p \deq \lnot G \equiv p \ndeq G.\]
  Thus $p \deq \lnot G \entails F$ (which matches
  Prop.~\ref{prop-char-sol-def}) holds if and only if $\lnot F \entails p \deq
  G$.

\item \label{solcase-def-nft} \textit{Each unknown has a ground instance that
  is definable in the negated formula.}  The sequence of negated definientia
  is a solution.  If $p_1(\tts_1) \ldots p_n(\tts_n)$ are the definable ground
  instances, then optionally in each solution component $G_i$, under the
  assumption $\clean{G_i}$, each member $t_{ij}$ of $\tts_i = t_{i1}\ldots
  t_{i\arity{p_i}}$ can be replaced by $x_j$.  The construction of the
  definientia can be performed with Craig interpolation, as described above
  for predicate definientia, except that an instance $p(\tts)$ takes the place
  of $p(\xxs)$. The difficulty is to find suitable instantiations $\tts_1
  \ldots \tts_n$.  A way to avoid guessing might be to let the formula whose
  proof serves as basis for interpolant extraction follow the schema \[\exists
  \ys\, (F \land p(\ys) \imp \lnot (F[p^{\prime}]\land \lnot p^\prime(\ys))),\]
  where $\ys = y_1 \ldots y_{\arity{p}}$, and take the instantiation of $\ys$
  found by the prover.  If the proof involves different instantiations of
  $\ys$ it has to be rejected. Rationale: Similar to
  the case~(\ref{solcase-def-nullary}.) since for ground atoms $p(\tts)$ it
  holds in general that $p(\tts) \equi G \equiv \lnot (p(\tts) \equi \lnot
  G)$.
\end{enumerate}

\noindent
These cases suggest to compute particular solutions based on
Prop.~\ref{prop-mse} by computing solutions for \UNSPs for each unknown, which
is inspected for matching the listed cases or other types of solvable cases,
for example the forms required by Schröder's interpolant or by Ackermann's
lemma. If that fails for an unknown, an attempt with the unknowns re-ordered
is made.  For propositional problems, an interpolating QBF solver would be a
candidate to compute solutions. Encodings of QBF into predicate logic, e.g.,
\cite{qbf:pl}, could possibly be applied for general first-order formulas with
nullary unknowns.

\subunit{The EHW-Combination of \ELIM-Witnesses for Disjuncts}
\label{sec-ehw}

Eberhard, Hetzl and Weller show in \cite{eberhard} that determining the
existence of a Boolean unifier (or, in, our terms, particular solution) for
quantifier-free predicate logic is $\mathrm{\Pi}^P_2$-complete, as for
propositional logic \cite{baader:boolean:98}.  Their proof rests on the
existence of an $\f{EXPTIME}$ function $\f{wit}$ from quantifier-free formulas
to quantifier-free formulas such that $\exists p F[p] \equiv
F[\f{wit}(F[p])]$.  The specification of $\f{wit}(F[p])$ is presented there as
a variation of the DLS algorithm \cite{dls,dls:conradie} for second-order
quantifier elimination: The input is converted to disjunctive normal form and
a specialization of Ackermann's lemma is applied separately to each
disjunct. The results for each disjunct are then combined in a specific way to
yield the overall witness formula.
The following proposition states a generalized variation of this technique that
is applicable also to other classes of inputs, beyond the quantifier-free
case.

\begin{prop}[EHW-Combination of \ELIM-Witnesses for Disjuncts]
\label{prop-ehw}
Let $F[p] = \bigvee_{i=1}^n F_i$ be a formula and let $G_1, \ldots, G_n$ be
formulas such that for $i \in \{1,\ldots n\}$ it holds that
$\ssubst{G_i}{p_i}{F_i}$ and $\exists p\, F_i[p] \equiv F_i[G_i]$.  Assume
that there are no free occurrences of $\XX$ in $F$ and, w.l.o.g, that no
members of $\free{F} \cup \XX$ are bound by a quantifier occurrence in $F$.
Let
\[G(\xs) \eqdef
\bigwedge_{i=1}^n ((\bigwedge_{j=1}^{i-1} \lnot F_j[G_j])
\land F_i[G_i] \imp G_i(\xs)).\]
Then it holds that $\ssubst{G}{p}{F}$ and $\exists p\, F[p] \equiv F[G].$
\end{prop}

\noindent
Formulas $G$ and $G_i$ are written as $G_i(\xs)$ and $G(\xs)$ where they occur
as formula constituents instead of substituents to emphasize that $\xs$ may
occur free in them.

\begin{proof}[Proposition~\ref{prop-ehw}]
This proof is an adaptation of the proof of Theorem~2 in \cite{eberhard}.  We
write here \name{$I$ is a model of $F$} symbolically as $I \mods F$.  That
$\ssubst{G}{p}{F}$ follows from the preconditions of the proposition and the
construction of $G$.  The right-to-left direction of the stated equivalence,
that is, \[\bigvee_{i=1}^{n} F_i[G]\; \entails\; \exists p\, \bigvee_{i=1}^{n}
F_i[p],\] then follows from Prop.~\ref{prop-fixing}.  The left-to-right
direction of the equivalence can be show in the following steps, explained
below.
\[
\begin{array}{r@{\hspace{1em}}l}
(1) & I \mods \exists p\, \bigvee_{i=1}^{n} F_i[p].\\
(2) & I \mods \bigvee_{i=1}^{n} \exists p\, F_i[p].\\
(3) & I \mods \bigvee_{i=1}^{n} F_i[G_i].\\
(4) & I \mods (\bigwedge_{j=1}^{k-1} \lnot F_j[G_j]) \land F_k[G_k].\\
(5) & I \mods \forall \xs\, (G(\xs) \equi G_k(\xs)).\\
(6) & I \mods F_k[G].\\
(7) & I \mods \bigvee_{i=1}^{n} F_i[G].
\end{array}
\]
Let $I$ be an interpretation such that (1) holds.  Step~(2) is equivalent
to~(1).  Assume the precondition of the proposition that for all $i \in
\{1,\ldots n\}$ it holds that $\exists p\, F_i[p] \equiv \bigvee_{i=1}^{n}
F_i[G_i]$. Step (3) follows from this and~(1).  By (3) there is a smallest
member $k$ of $\{1,\ldots,n\}$ such that $I \models F_k[G_k]$. This
implies~(4).  The left-to-right direction of the equivalence in~(5) follows
since if $I \mods G(\xs)$ then by~(4) and the definition of $G(\xs)$ it is
immediate that $I \mods G_k(\xs)$.
The right-to-left direction of the equivalence in~(5) can be shown as follows.
Assume $I \mods G_k(\xs)$.  Then $I$ is a model of the $k$th conjunct of
$G(\xs)$ since $G_k(\xs)$ is in the conclusion of that conjunct, and $I$ is a
model of each $j$th conjunct of $G(\xs)$ with $j \neq k$, because the
antecedent of such a conjunct contradicts with~(4).  Step~(6) follows from~(4)
and~(5). Step~(7) follows from~(6).
\qed
\end{proof}

The following proposition is another variation of the EHW-combination of
witnesses for disjuncts. It can be proven similarly to
Prop.~\ref{prop-ehw}.

\begin{prop}[Alternate Variation of EHW-Combination]
\label{prop-disj-witness-general} Let $F[p]$ be a \UNSP and let $G_1, \ldots,
G_n$ be formulas such that for $i \in \{1,\ldots,n\}$ it holds that
$\ssubst{G_i}{p}{F}$ and such that $\exists p\, F \equiv \bigvee_{i=1}^n
F[G_i].$ Assume that there are no free occurrences of $\{x_i\mid i \geq 1\}$
in $F$ and, w.l.o.g, that no members of $\free{F} \cup \xs$ are bound in by a
quantifier occurrence in $F$.  Let
\[G \eqdef
\bigwedge_{i=1}^n ((\bigwedge_{j=1}^{i-1} \lnot F[G_j]) \land F[G_i] \imp G_i).\]
Then $\ssubst{G}{p}{F}$ and $\exists p\, F \equiv F[G]$.
\end{prop}

Proposition~\ref{prop-disj-witness-general} is also applicable to \UNSPs of
the form handled by Prop.~\ref{prop-ehw}, but for this case leads to a more
clumsy result~$G$: Assume the additional precondition that for all $j \in
\{1,\ldots,n\}$ it holds that $\ssubst{G_j}{p}{\bigvee_{i=1}^{n} F_i}$.  Let
$F[p] \eqdef \bigvee_{i=1}^{n} F_i[p]$.  Then $\exists p F[p] \equiv F_1[G_1]
\lor \ldots \lor F_n[G_n] \equiv F[G_1] \lor \ldots \lor F[G_n]$.

\subunit{Relational Monadic Formulas and Relaxed Substitutibility}

The class of relational monadic formulas with equality, called here \MONE, is
the class of first-order formulas with equality, with unary predicates and
with individual constants but no other functions (without equality it is the
\defname{Löwenheim class}). It is decidable and permits second-order
quantifier elimination, that is, each formula in \MONE extended by predicate
quantification is equivalent to a formula in \MONE.  As shown in
\cite{cw-relmon} it has interesting relationships with $\mathcal{ALC}$.
Behmann \cite{beh:22} gave a decision method for \MONE that performs
second-order quantifier elimination by equivalence-preserving formula
rewriting \cite{cw-relmon,cw-behmann}. Almost three decades later he published
an adaptation of these techniques to the solution problem for
\name{Klassenlogik}
\cite{beh:50:aufloesungs:phil:1,beh:51:aufloesungs:phil:2}, which in essence
seems to be \MONE. It still remains open to assess this and apparently related
works by Löwenheim \cite{loewenheim:1910}.

Under a relaxed notion of substitutibility, the construction of \ELIM-witnesses
for \MONE is possible by joining Behmann's rewriting technique \cite{beh:22}
with the EHW-combination (Prop.~\ref{prop-ehw}).
Let $F[p]$ be a \MONE formula and let $p$ be a unary predicate.  Assume that
$\free{F} \cap \XX = \emptyset$. The reconstruction of Behmann's normalization
shown in the proofs of Lemma~14 and Lemma~16 of \cite{cw-behmann} can be
slightly modified to construct a formula $F^\prime = \bigvee_{i = 1}^{n}
F^{\prime\prime}_i$ that is equivalent to $\exists p\, F$ and such that each
$F^{\prime\prime}_i$ is of the form
\begin{equation}
F^{\prime\prime}_i = C_i \land \exists \uus_i\, (D_i(\uus_i) \land \exists p\,
(\forall y\, (A_i(\uus_iy) \imp p(y)) \land \forall y\, (p(y) \imp
B_i(\uus_iy)))),\end{equation} where $\uus_i$ is a sequence of individual
symbols such that $\uus_i \cap \free{C_i} = \emptyset$, predicate~$p$ has only
the two indicated occurrences and $\free{F^\prime_i} \subseteq \free{\exists
  p\, F}$.  Let
\begin{equation}F^{\prime\prime\prime}_i[p] = C_i \land D_i(\uus_i) \land \forall y\,
(A_i(\uus_iy) \imp p(y)) \land \forall y\, (p(y) \imp
  B_i(\uus_iy)).\end{equation} Then $F^{\prime\prime}_i \equiv \exists \uus_i
\exists p\, F^{\prime\prime\prime}_i \equiv \exists \uus_i
F^{\prime\prime}_i[A(\uus_i x_1)]$, where the last equivalence follows from
Ackermann's lemma (Prop.~\ref{prop-ackermann}).  It holds that
$\ssubst{F^{\prime\prime}_i}{p}{A(\uus_i x_1)}$ but, since the quantified
symbols $\uus$ may occur in $A(\uus_i x_1)$, the substitutibility condition
$\ssubst{\exists \uus\, F^{\prime\prime}_i}{p}{A(\uus_i x_1)}$ does not hold
in general.
The variables $\uus_i$ can be gathered to a single global prefix $\uus$
(assuming w.l.o.g. that none of them occurs free in any of the $C_i$) such
that $F \equiv \exists \uus\, F^{\prime\prime\prime\prime}[p]$ where
$F^{\prime\prime\prime\prime} = \bigvee_{i=1}^n F^{\prime\prime\prime}_i$.  By
Prop.~\ref{prop-ehw} we can construct a formula~$G(\uus)$ such that
$\ssubst{G(\uus)}{p}{F^{\prime\prime\prime\prime}}$ and $\exists p\,
F^{\prime\prime\prime\prime}[p] \equiv
F^{\prime\prime\prime\prime}[G(\uus)]$. This implies $\exists p\, F[p] \equiv
F[G(\uus)]$. However, substitutibility of $G(\uus)$ holds only with respect to
$F^{\prime\prime\prime\prime}$, while $\ssubst{G(\uus)}{p}{F}$ does not hold
in general.
Thus, under a relaxed notion of substitutibility that permits the
existentially quantified $\uus$ in the witness, the EHW-combination can be
applied to construct witnesses for \MONE formulas.

\section{Conclusion}
\label{sec-conclusion}

The \name{solution problem} and second-order quantifier \name{elimination}
were interrelated tools in the early mathematical logic.  Today elimination
has entered automatization with applications in the computation of
circumscription, in modal logics, and for semantic forgetting and modularizing
knowledge bases, in particular for description logics.  Since the solution
problem on the basis of first-order logic is, like first-order validity,
recursively enumerable there seems some hope to adapt techniques from
first-order theorem proving.

The paper makes the relevant scenario accessible from the perspective of
predicate logic and theorem proving.  It shows that a wealth of classical
material on Boolean equation solving can be transferred to predicate logic and
only few essential diverging points crystallize, like the constructability of
witness formulas for quantified predicates, and ``Schröder's reproductive
interpolant'' that does not apply in general to first-order logic.  An
abstracted version of the core property underlying the classical method of
successive eliminations provides a foundation for systematizing and
generalizing algorithms that reduce $n$-ary solution problems to unary
solution problems.  Special cases based on Craig interpolation have been
identified as first steps towards methods for solution construction.

Beyond the presented core framework there seem to be many results from
different communities that are potentially relevant for further
investigation. This includes the vast amount of techniques for equation
solving on the basis of Boolean algebra and its variations, developed over the
last 150 years.  For description logics there are several results on concept
unification, e.g., \cite{baader:narendran:01,baader:morawska:10}.  Variations of
Craig interpolation such as disjunctive interpolation \cite{ruemmer:disjipol}
share with the solution problem at least the objective to find substitution
formulas such that the overall formula becomes valid (or, dually,
unsatisfiable).

Among the issues that immediately suggest themselves for further research are
the parallel between nondeterministic methods with execution paths for each
particular solution and methods that compute a most general solution, the
exploration of formula simplifications and techniques such as definitional
normal forms to make constructions like \name{rigorous solution} and
\name{reproductive interpolant} feasible, and the investigation of the relaxed
notion of substitutibility under which solutions for relational monadic
formulas can be constructed.  The possible characterization of \name{solution}
by an entailment also brings up the question whether Skolemization and
Herbrand's theorem justify some ``instance-based'' technique for computing
solutions that succeeds on large enough quantifier expansions.

\subsection*{Acknowledgments}
The author thanks anonymous reviewers for their helpful comments.
Funded by the Deutsche Forschungsgemeinschaft (DFG, German
Research Foundation) -- Project-ID~264466967 and~457292495.

\markboth{\small \hfill}{\small References \hfill}
\renewcommand{\stitle}{References}
\phantomsection
\addcontentsline{toc}{section}{References}
\bibliographystyle{splncs04} \bibliography{bibelim06short}

\end{document}